%% file: DRSTL.tex
\documentclass[letterpaper, 10 pt]{ieeetran} 
\IEEEoverridecommandlockouts                              
\usepackage{graphics} 
\usepackage{epsfig} 
\usepackage{amsmath} 
\usepackage{amssymb}  
 \usepackage{enumitem}

\usepackage{changes}

\usepackage{mathtools,amsthm}
\allowdisplaybreaks

\usepackage{url}
\usepackage{hyperref}
\usepackage{graphics} 
\usepackage{epsfig}
\usepackage{times} 
\usepackage{amsmath,amssymb,amsfonts}
\usepackage{algorithmic}
\usepackage{graphicx}
\usepackage{textcomp}
\usepackage{varioref}
\usepackage{import}
\usepackage{framed}
\usepackage{comment}
\usepackage{epstopdf}   
\usepackage{psfrag}
\usepackage{breqn}
\usepackage{float}
\usepackage{booktabs}
\usepackage{soul}
\usepackage{amsbsy}
\usepackage[bottom]{footmisc}
\usepackage{lineno}
\usepackage{cleveref}
\usepackage{microtype}
\usepackage{comment}
\usepackage{arydshln}
\usepackage{mathtools}
\usepackage{tikz}
\newtheorem{Theorem}{Theorem}
\newtheorem{Proposition}{Proposition}
\newtheorem{Corollary}{Corollary}
\newtheorem{Lemma}{Lemma}

\newtheorem{Remark}{Remark}

\newtheorem{Assumption}{Assumption}

\usetikzlibrary{patterns}
\usepackage{todonotes}

\newcommand{\vect}[1]{\ensuremath{\boldsymbol{\mathrm{#1}}}}

\def\r{\rho}
\def\f{\varphi}
\DeclareMathOperator{\G}{\Box}
\DeclareMathOperator{\F}{\rotatebox[origin=c]{45}{$\Box$}}
\DeclareMathOperator*{\argmin}{argmin}

\usepackage{cite}
\usepackage{amsmath,amssymb,amsfonts}
\usepackage{algorithmic}
\usepackage{graphicx}
\usepackage{algorithm,algorithmic}
\usepackage{hyperref}
\usepackage{textcomp}
\def\BibTeX{{\rm B\kern-.05em{\sc i\kern-.025em b}\kern-.08em
    T\kern-.1667em\lower.7ex\hbox{E}\kern-.125emX}}

\begin{document}

\title{Data-Driven Distributionally Robust Control for Interacting Agents under Logical Constraints}

\author{Arash Bahari Kordabad, Eleftherios E. Vlahakis, Lars Lindemann, Sebastien Gros, Dimos V. Dimarogonas, and Sadegh Soudjani
\thanks{Arash Bahari Kordabad and Sadegh Soudjani are with the Max Planck Institute for Software Systems, Kaiserslautern, Germany. E-mail: {\tt\small\{arashbk, sadegh\}@mpi-sws.org.} Eleftherios E. Vlahakis and  Dimos V. Dimarogonas are with the Division of Decision and Control Systems, KTH Royal Institute of Technology, Stockholm, Sweden. E-mail: {\tt\small\{vlahakis, dimos\}@kth.se}.  Lars Lindemann is with the Thomas Lord Department of Computer Science and the Ming Hsieh Department of Electrical and Computer Engineering, Viterbi School of Engineering, University of Southern California, Los Angeles, USA. 
Email: {\tt\small llindema@usc.edu}. Sebastien Gros is with the Department of Engineering Cybernetics, Norwegian University of Science and Technology (NTNU), Trondheim, Norway. E-mail: {\tt\small  sebastien.gros@ntnu.no}. \newline This research is supported by the following grants: EIC 101070802 and ERC 101089047.}}

\maketitle
\begin{abstract}
In this paper, we propose a distributionally robust control synthesis for an agent with stochastic dynamics that interacts with other agents under uncertainties and constraints expressed by signal temporal logic (STL). We formulate the control synthesis as a chance-constrained program (CCP) with STL specifications that must be satisfied with high probability under all uncertainty tubes induced by the other agents. To tackle the CCP, we propose two methods based on concentration of measure (CoM) theory and conditional value at risk (CVaR) and compare the required assumptions and resulting optimizations. These approaches convert the CCP into an expectation-constrained program (ECP), which is simpler to solve than the original CCP. To estimate the expectation using a finite set of observed data, we adopt a distributionally robust optimization (DRO) approach. The underlying DRO can be approximated as a robust data-driven optimization that provides a probabilistic under-approximation to the original ECP, where the probability depends on the number of samples. Therefore, under feasibility, the original STL constraints are satisfied with two layers of designed confidence: the confidence of the chance constraint and the confidence of the approximated data-driven optimization, which depends on the number of samples. We then provide details on solving the resulting robust data-driven optimization numerically. Finally, we compare the two proposed approaches through case studies.
\end{abstract}
\section{Introduction}

In the rapidly evolving landscape of autonomous systems, designing control strategies for systems that interact with other systems in a shared environment is a critical challenge \cite{dennis2023verifiable,lindemann2025formal}. 
Consider a car, acting as a controlled agent, crossing an intersection in the presence of other cars that act as uncontrollable agents.
The controlled agent must account for all possible behaviors of other agents in the environment. From the perspective of the controlled agent, the other agents can be treated as uncontrollable agents with uncertainties in their behavior.  Figure~\ref{fig:0} illustrates this example.

\begin{figure}
\centering
\includegraphics[width=0.4\textwidth]{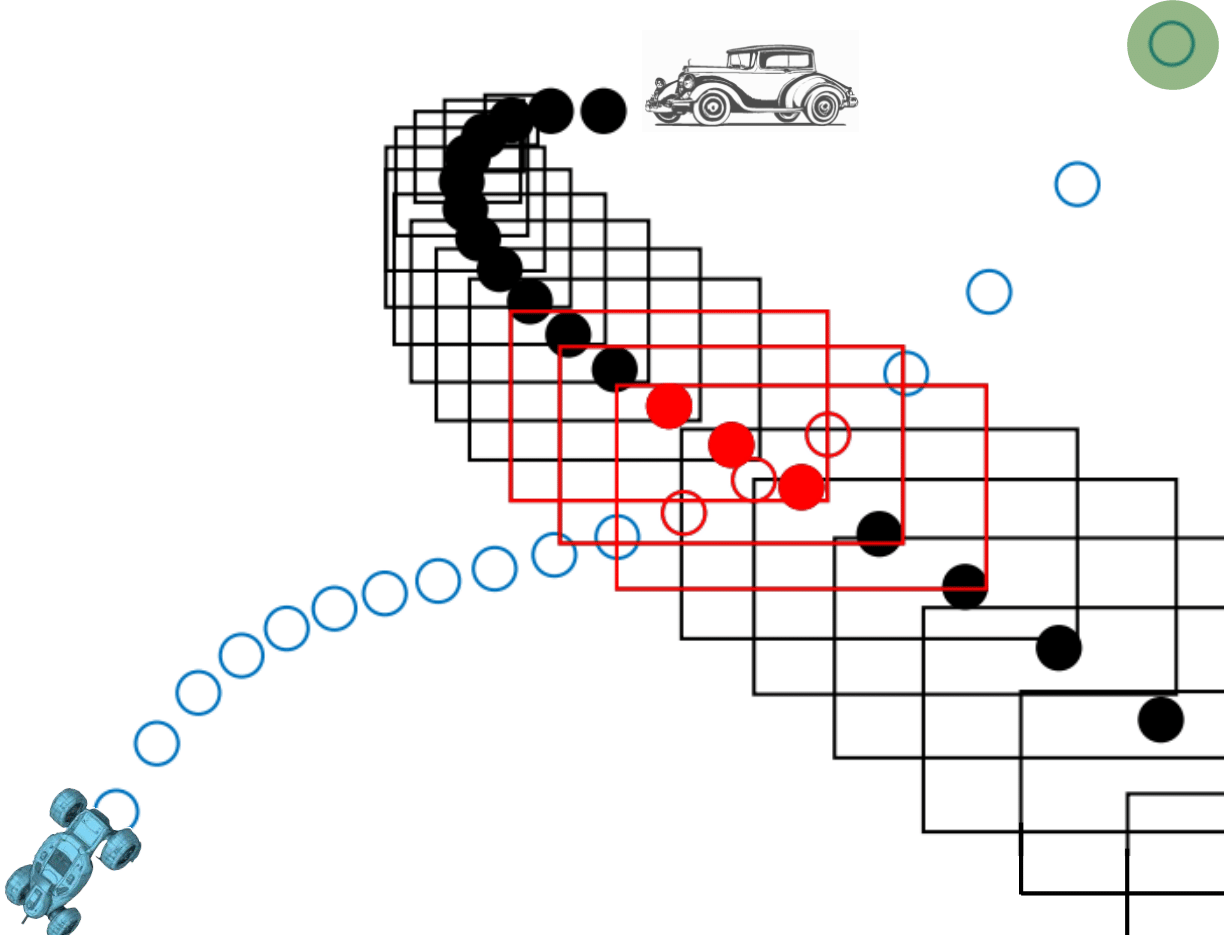}
\caption{The controlled agent in blue, aims to reach the location indicated by a green circle while avoiding collisions with the other agent. The other agent in black is modeled as an uncontrollable agent with an uncertainty tube from the perspective of the controlled agent. The red color indicates collisions along the path of the controlled agent indicated by blue circles.}
\label{fig:0}
\end{figure}

 A primary objective in this setting is often to achieve an individual goal for the controlled agent, while also accounting for interaction constraints with other uncontrollable agents, such as ensuring collision avoidance in shared environments. A formal framework for task specification and verification is thus crucial in such applications, to ensure that the agent's behavior aligns with the desired objectives while respecting constraints. In this paper, we consider these interaction constraints to be expressed as signal temporal logic (STL) specifications \cite{maler2004monitoring}. 

STL is a formal language that allows us to encode spatio-temporal constraints~\cite{maler2004monitoring}. The building blocks of STL include predicates, predicate functions, and logical and temporal operators. The predicate functions are defined over system trajectories and assign a real value to assess STL semantics. Building upon these components, one can precisely specify a wide range of complex spatio-temporal properties of a dynamical system. STL is equipped with quantitative semantics~\cite{donze2010robust}, also known as robustness function, which quantifies the robustness of satisfaction of the specifications at hand and can be used for efficient control synthesis. In a deterministic framework, it is possible to synthesize sound and complete algorithms for STL control~\cite{raman-cdc14}. However, when dynamical systems are subject to stochastic disturbances and STL constraints, the synthesis problem can be formulated as a stochastic optimal control problem, in which the formal specification can be formulated via probabilistic constraints, enabling a systematic approach to quantifying uncertainty and handling feasibility~\cite{farahani2017shrinking, Haesaert2018}. The resulting chance-constrained program (CCP) is challenging to solve due to the complexity of the robustness function of the STL constraints, which grows with the number of involved agents.

\smallskip
\noindent\textbf{Related works.} In a probabilistic setting, there are two main approaches to defining STL constraints. Many recent results pertain to the application of probability or risk measures to individual predicates, while several works extend the application of probability measures to the entire STL specification itself. To address the impact of critical tail events when STL formulas are violated,~\cite{Safaoui2020,LindemannCDC2020} propose risk STL, incorporating risk constraints over predicates while preserving Boolean and temporal operators. Authors of~\cite{Sadigh2016} introduce probabilistic STL, allowing expression of uncertainty by incorporating random variables into predicates. Similarly, in~\cite{Sadigh2018}, chance-constrained temporal logic formulates chance constraints as predicates to model perception uncertainty for autonomous vehicles. Stochastic temporal logic introduced in~\cite{Li2017,Kyriakis2019} is similar in syntax to chance-constrained temporal logic but is designed for stochastic systems, where perturbations affect system dynamics rather than predicate coefficients. Another approach for stochastic systems is presented in~\cite{Tiger2020}, extending STL with a sub-language for probabilities, observations, and predictions. The authors in~\cite{Ilyes2023} introduce the stochastic robustness measure that quantifies how well a stochastic trajectory satisfies a specification.

The recent top-down approaches that study STL probabilistic verification of stochastic systems considering chance constraints on the entire specification are in~\cite{ScherHSCC2022, Scher2022}. More closely related to our work is~\cite{FarahaniTAC2019}, in which the authors transform chance constraints into linear constraints using the concentration of measure (CoM) inequalities to provide a conservative approximation of the feasible domain. Robust techniques for STL synthesis have also been studied~\cite{SadBel:15,Farahani2015, RamanHSCC2015}, accounting for the most unfavorable variation of the underlying random parameters. Due to the nonconvex feasible domains typically induced by CCPs, many works pertain to the study of numerical methods to handle CCPs, such as randomized optimization in which the original optimization is approximated by a scenario program (SP) by sampling the uncertainty space. SP approaches have been studied for convex~\cite{Campi2011, Calafiore2010} and nonconvex~\cite{Grammatico2016} CCPs. 

Unlike stochastic optimization settings where the probability distribution is assumed to be known, \textit{distributionally robust optimization} (DRO) addresses the lack of information on the probability distribution by considering the worst-case distribution within an ambiguity set. Various methods exist for constructing ambiguity sets, such as moment ambiguity~\cite{delage2010distributionally}, Kullback–Leibler divergence-based ball~\cite{hu2013kullback}, and Wasserstein-based ball~\cite{pflug2007ambiguity}. The Wasserstein ambiguity set represents a statistical ball within the space of probability distributions surrounding the empirical distribution, with its radius measured using Wasserstein distance. Wasserstein DRO offers probabilistic guarantees based on finite samples via a tractable formulation~\cite{mohajerin2018data} and has attracted significant attention recently~\cite{kordabad2022safe,gao2023distributionally}.

In the DRO literature, several works have addressed CCP directly. An explicit reformulation for both individual and joint CCPs has been presented in~\cite{chen2024data}, where uncertainties are modeled as affine functions. The authors of~\cite{xie2021distributionally} reformulated CCP as a conditional value-at-risk (CVaR) mixed-integer program for affine functions. A more general approach for CCP has been proposed in~\cite{hota2019data}. However, dealing with expectation-constrained programs (ECPs) preserves linearity and convexity, making it often computationally simpler and more straightforward compared to evaluating or approximating CCPs, particularly for non-standard distributions.

In~\cite{Soudjani2018}, the authors propose a method to approximate a general CCP as an ECP using CoM property. More specifically, the approach solves CCP by requiring sampling-based constraints to be satisfied on average rather than individually, based on the assumption that the underlying distribution of the uncertain variables satisfies a CoM~\cite{Ledoux1999} property and exhibits bounded variance. 
The CoM property implies that if a set in a probability space has a measure of at least one-half, then most points in the probability space are close to the set. Equivalently, if a function defined on the probability space is regular enough, the chance of it deviating significantly from its expectation or median is small. When dealing with stochastic systems under uncertainty featuring the CoM property, one can anticipate these systems to display regular and predictable dynamic behavior.

Another approach for quantifying the influence of stochasticity on the system is by defining appropriate risk measures. For example, the value at risk (VaR) of a random variable provides a threshold value below which a certain percentage of outcomes will be concentrated. However, since VaR is often difficult to optimize (except for normal distributions), it is less practical compared to the conditional value at risk (CVaR), which has favorable mathematical properties such as convexity and monotonicity~\cite{rockafellar2000optimization}. CVaR quantifies the expected value of a random variable, representing, e.g., a cost, given that an outcome exceeds the VaR threshold. CVaR can be cast as the tight upper bound of VaR and is essentially the average of the worst outcome beyond a certain confidence level.

A framework is proposed in~\cite{lindemann2023safe} for performing probabilistic safe planning on deterministic systems operating in unknown, stochastic dynamic environments using conformal prediction. Leveraging the concept of conformal prediction, control synthesis for a dynamical system among uncontrollable agents has been investigated in~\cite{yu2023signal}. The control problem has been formulated as a worst-case bilevel mixed-integer program (MIP) using conformal prediction and probabilistic prediction regions.
In contrast, our proposed method assumes such prediction regions are available to the controlled agent in the form of uncertainty tubes and provides a data-driven safe control synthesis for a stochastic agent among other uncontrollable uncertain agents. Moreover, the challenges associated with MIP are transferred to the robustness function of STL specifications, and the chance constraint directly provides the confidence level of satisfaction.

\smallskip
\noindent\textbf{Contributions.} In this paper, we consider the control synthesis problem of a controlled stochastic agent interacting with other uncontrollable agents in a shared environment. The goal is to optimize a certain objective function for the controlled agent while satisfying a probabilistic STL specification under worst-case behavior of the uncontrollable agents by solving a stochastic optimal control problem. We then formulate this stochastic optimal control problem as a CCP, where the underlying system is subject to stochastic disturbances with unknown distribution and a set of STL specifications under uncertainty. We first assume that the disturbance distribution is consistent with the CoM property and that the underlying predicate functions involved in the STL constraints are Lipschitz continuous functions. We then transform the CCP into an ECP. Since the exact distribution is assumed to be unknown, we propose a data-driven Wasserstein distributionally robust approach that guarantees STL satisfaction with respect to the exact distribution with high confidence. Moreover, CVaR is employed as an alternative method to transform the CCP into an ECP. We then compare the resulting distributionally robust optimization with the one from the CoM approach.

The current paper extends the preliminary results presented in the conference paper~\cite{kordabad2024distributionally} in the following directions: 
(a) we address controlled agents in a dynamically changing environment while \cite{kordabad2024distributionally} is limited to static environments;
(b) we propose to use CVaR as an alternative approach for tackling CCPs;
(c) we make the computations robust to uncertainty tubes in the behavior of uncontrollable agents and provide a comparison between the results based on CoM and CVaR through detailed case studies; and
(d) we provide the proofs of statements with additional numerical examples to illustrate the proposed techniques.

\smallskip
\noindent\textbf{Outline.} The remainder of the paper is organized as follows. Section~\ref{sec:Prelim} provides preliminaries on dynamical systems and STL specifications. Section~\ref{sec:chance} details the resulting robust CCP and proposes two methods including CoM and CVaR, to transform the CCP into ECPs. Section~\ref{sec:ecp} presents a distributionally robust expression of the proposed methods and corresponding data-driven optimization. Section~\ref{sec:rp} provides details of solving the resulting data-driven optimization. Numerical studies are presented in Section~\ref{sec:examples}, illustrating the efficacy of the approach and comparing the proposed methods, while Section \ref{sec:concl} provides a conclusion of the paper and discusses potential directions for future work.
 
\section{Preliminaries}\label{sec:Prelim}
In this section, we provide preliminaries for the considered system of interacting agents and STL specification. 

\subsection{Dynamics}\label{sec:dyn}
In the following, we describe the considered setting that includes one controlled agent interacting with $n\geq 1$ uncontrollable agents.

We consider the controlled agent evolving in discrete time with state space $\mathcal{X}\subseteq\mathbb{R}^{n_x}$, input space $\mathcal{U}\subseteq\mathbb{R}^{n_u}$, and disturbance set $\mathcal{W}\subseteq\mathbb{R}^{n_x}$, with $n_x$ being the dimension of the state and disturbance vectors and $n_u$ being the dimension of the input vector. The dynamics of the agent can be modeled by the following linear difference equations perturbed by stochastic disturbances:
\begin{equation}\label{eq:discsys1}
    x_{k+1}=Ax_k+\sum_{i=1}^n B_iy^i_k+Cu_k+w_k,
\end{equation}
where the subscript $k\in \mathbb{N}_0$ is the time index with $\mathbb{N}_0$ being the set of non-negative integers, $x_k\in\mathcal{X}$ denotes the state of the controlled agent, $y^i_k\in\mathbb{R}^{n_i}$ denotes the state of the uncontrollable agent $i$, $1\leq i\leq n$, $u_k\in\mathcal{U}$ denotes the control input at time instant $k$, and $w_k\in\mathcal{W}$ is a vector of random variables at time~$k$, with a probability distribution $\mathcal{P}$ supported on $\mathcal{W}$. Matrices $A\in\mathbb{R}^{n_x\times n_x}$, $B_i\in\mathbb{R}^{n_x\times n_i}$, and $C\in\mathbb{R}^{n_x\times n_u}$ are known system matrices, and the initial state $x_0$ is assumed to be known. For any $k\in\mathbb N$, the state space model in~\eqref{eq:discsys1} provides the following explicit form for $x_k$, as a function of $x_0$, the state sequence of other agents ${y}^i_{0:k}:=[{y_0^i}^\intercal,\ldots,{y^i_{k-1}}^\intercal]^\intercal$, input sequence vector ${u}_{0:k}:=[u^\intercal_0,\ldots,u^\intercal_{k-1}]^\intercal$, 
and the disturbance sequence ${w}_{0:k}:=[w^\intercal_0,\ldots,w^\intercal_{k-1}]^\intercal$:
\begin{equation}\label{eq:traj1}
x_{k} = A^{k} x_0 + \sum_{l=0}^{k-1}A^{k-l-1} \left(\sum_{i=1}^{n} B_i y^i_l+C u_l + w_l\right).
\end{equation}
The uncontrollable agents are modeled with linear dynamics that depend on the state of the controlled agent, the states of other uncontrollable agents, and may include some uncertainties:
\begin{equation*}
y^i_{k+1}= D_ix_k+\sum_{j=1}^n E_{i,j} y^j_{k}+\theta^i_k,\qquad \forall i\in\{1,2,\ldots, n\}, 
\end{equation*}
where $D_i\in \mathbb{R}^{n_i\times n_x}$ and $E_{i,j}\in \mathbb{R}^{n_i\times n_j}$, $\forall i\in\{1,2,\ldots, n\}$,  are known matrices, and $\theta^i_k\in \mathcal{T}_i\subset \mathbb{R}^{n_i}$, $\forall i\in\{1,2,\ldots, n\}$, $\forall k\geq 0$, are uncertain variables that the controlled agent is not aware of, and belong to the compact set $\mathcal{T}_i$. The set $\mathcal{T}_i$ is assumed to be known for the controlled agent. Similarly, the state $y^i_k$ can be written as follows:
\begin{equation}\label{eq:traj2}
y^i_{k}=  E^k_{i,i} y^i_{0}+ \sum_{l=0}^{k-1}E_{i,i}^{k-l-1} \left(\sum_{\substack{j=1\\ j\neq i}}^{n} E_{i,j}  y^j_{l}+D_ix_l+\theta^i_l\right),
\end{equation}
for all $k\in\mathbb N$. The augmented state of all agents, including the controlled and the uncontrollable agents, can be denoted by $z_{k}=[x^\intercal_k,{y_k^1}^\intercal, \ldots, {y_k^n}^\intercal]^\intercal\in \mathbb{R}^{n_z}$, with $n_z=n_x+\sum_{i=1}^n n_i$, and we have:
\begin{equation}\label{eq:dyn:ag}
    z_{k+1}\!\!=\!\!\underbrace{\left[\begin{array}{ c | c c c }
    A & B_1 & \cdots & B_n  \\
    \hline
 D_1 &  &  & \\
  \vdots &  & [E] & \\
   D_n &  &  & 
  \end{array}\right]}_{=:\bar A} \!\!z_k\!+\!\! \left[\begin{array}{ c }
    C  \\
    \hline
    0 
  \end{array}\right]\!\!u_k\!+\!\!\left[\begin{array}{ c }
    w_k  \\ 
    \hline
    \theta^1_k \\ \vdots \\ \theta^n_k
  \end{array}\right]\!,
\end{equation}
where matrix $E$ is a block matrix with the elements $E_{i,j}$ located in the $i$-th row block and $j$-th column block. Similar to \eqref{eq:traj1} and \eqref{eq:traj2}, the augmented state $z_{k}$ can be expressed linearly as a function of the initial states, $x_0$, $y^i_0$, $\forall i\in\{1,\ldots,n\}$, input sequence $u_{0:k}$, disturbance $w_{0:k}$ and uncertain variables $\theta^i_{0:k}$, $\forall i\in\{1,\ldots,n\}$, where ${\theta}^i_{0:k}:=[{\theta_0^i}^\intercal,\ldots,{\theta_{k-1}^i}^\intercal]^\intercal$. This expression is linear for agents with linear dynamics. We will utilize this observation in this paper.

\begin{Remark} 
The variables $\theta_k^i$ in \eqref{eq:traj2} induce an \emph{uncertainty tube} for the variables $y_k^i$. Such tubes can be constructed with certain confidence using data from the agents and approaches such as conformal prediction \cite{lindemann2023safe,yu2023signal}.
In this work, we have assumed that such uncertainty tubes are already available to the controlled agent. Any confidence attached to these uncertainty tubes can be combined a posterior with the probabilistic guarantees provided by our approach.
\end{Remark}

The goal is to synthesize a controller for the controlled agent to accomplish tasks expressed in signal temporal logic (STL) while interacting with uncontrollable agents. In the following, the preliminaries for STL are presented.

\subsection{STL specifications}\label{sec:stl}
A finite run of system \eqref{eq:dyn:ag} can be considered as a signal 
$\xi = \{z_0,z_1,z_2,\dots, z_N\}$, which is a sequence of observed augmented states with finite horizon $N$.
We consider STL formulas defined recursively according to the grammar~\cite{MalNic:04}: 
\begin{equation*}
\varphi ::= \top\mid \pi \mid \neg \varphi \mid\varphi \land \psi  \mid \varphi\, {U}_{[a,b]}\,\psi,
\end{equation*}
where $ \top$ is the \emph{true} predicate, $\pi$ is a predicate whose truth value is determined by the sign
of a corresponding predicate function of state variables, i.e., $\pi = \{\alpha(z)\ge 0\}$ with $\alpha:\mathbb R^{n_z}\rightarrow\mathbb R$, $\varphi$ and 
$\psi$ are STL formulas,
$\neg$ and $\land$ indicate \emph{negation} and \emph{conjunction} of formulas, and ${ U}_{[a,b]}$ is the \emph{until} operator that operates withing the time interval $[a,\, b]$ with $0\leq a\leq b$. A run $\xi$ satisfies $\varphi$ at time $k$, 
denoted by $(\xi,k) \models \varphi$, if the sequence 
$z_kz_{k+1}\ldots z_N$ satisfies $\varphi$. We say 
$\xi$ satisfies $\varphi$ if
$(\xi,0) \models \varphi$. Boolean semantics of STL formulas are defined as follows:
  \begin{flalign*}
        & ( \xi,k)\models \pi \quad\quad\quad\, \Leftrightarrow \quad  \alpha( z_k)\geq 0,\\
        & ( \xi,k)\models \lnot \varphi \quad \quad\,\,\, \Leftrightarrow \quad  \lnot(( \xi,k)\models \varphi),
        \\
        & ( \xi,k)\models \varphi \land \psi \quad\,\,  \Leftrightarrow \quad  ( \xi,k)\models\varphi \wedge ( \xi,k)\models\psi, \\
        & ( \xi,k)\models \varphi\, {U}_{[a,b]}\, \psi  \Leftrightarrow \quad \exists k'\in\{a,\ldots, b\}, ( \xi,k+k')\models \psi  \\ &\qquad \, \qquad\qquad \qquad\quad \wedge\forall k''\in\{0, \ldots, k'\}, ( \xi,k+k'')\models\varphi.
    \end{flalign*}
Additionally, we derive the \emph{disjunction} operator as $\varphi \lor \psi:=\neg(\neg\varphi\land\neg\psi)$,
the \emph{eventually} operator as $\F_{[a,b]}\varphi := \top{ U}_{[a,b]} \varphi$,
and the \emph{always} operator as $\G_{[a,b]}\varphi:=\neg \F_{[a,b]}\neg\varphi$. Thus $(\xi,k) \models \F_{[a,b]} \varphi$ if $\varphi$ holds at some time instant between $a+k$ and $b+k$ and
$(\xi,k) \models \G_{[a,b]} \varphi$ if $\varphi$ holds
at every time instant between $a+k$ and $b+k$.

\smallskip\noindent\textbf{STL Robustness.}
In addition to the above Boolean semantics, the quantitative semantics (a.k.a. robustness function) of STL \cite{req_mining_hscc2013} assigns to each formula $\f$ a real-valued function $\r^\f$ of signal $\xi$ and $k$ such that $\r^\f> 0$ implies $(\xi,k) \models \f$. The robustness of a formula $\varphi$ with respect to a run $\xi$ at time $k$ is defined recursively as
\begin{equation*}
\renewcommand{\arraystretch}{1.2}
\begin{array}{ll}
    \rho^{\top}(\xi,k)                 &= +\infty \\
    \rho^{\mu}(\xi,k)                  &= \alpha(z_k) \\
    \rho^{\lnot\varphi}(\xi,k)            &= -\rho^{\varphi}(\xi,k) \\
    \rho^{\varphi \wedge \psi}(\xi,k)     &= \min(\rho^{\varphi}(\xi,k), \rho^{\psi}(\xi,k)) \\
    \rho^{\varphi\, {U}_{[a,b]}\, \psi}(\xi,k) &= \max_{k' \in \{a,\ldots, b\}} \Big( \min\big( \rho^{\psi}(\xi,k+k'), \\
                                           &\quad \min_{k'' \in \{0,\ldots,k'\}} \rho^{\varphi}(\xi,k+k'') \big) \Big).
\end{array}
\end{equation*}
The value of the robustness function $\rho^{\varphi}(\xi,k)$ can be interpreted as how robustly the trajectory $\xi$ satisfies a given STL formula $\varphi$. The robustness of the formulas $\F_{[a,b]}\varphi$ and $\G_{[a,b]}\varphi$ are
\begin{flalign*}
    \r^{\F_{[a,b]} \varphi}(\xi,k)= & \max_{k'\in\{a,\ldots, b\}} \r^{\varphi}(\xi,k+k'),\nonumber\\
    \r^{\G_{[a,b]} \varphi}(\xi,k) =&\min_{k'\in \{a,\ldots, b\}}\r^{\varphi}(\xi,k+k').\nonumber
\end{flalign*}

As described, the robustness function is commonly defined by the arguments of the system trajectory and the time. However, for the scope of this study, it is more comprehensible to explicitly define the robustness as a function of input, disturbance, initial state, and time, incorporating the system dynamics. In essence, as discussed earlier, by accounting for the dynamics, the system trajectory is determined by the initial state, input, and disturbance sequence, and the mapping $(u_{0:N},w_{0:N},\theta^{1:n}_{0:N},z_0)\rightarrow \xi$ is linear. Therefore, for the sake of clarity, we define a dynamics-dependent function $\varrho^{\varphi}$ as follows:
\begin{equation*}
     \varrho^{\varphi} (u_{k:N},w_{k:N},\theta^{1:n}_{k:N},z_k,k):=\rho^{\varphi}(\xi,k),
\end{equation*}
where ${\theta}^{1:n}_{0:k}:=[{\theta_{0:k}^1}^\intercal,\ldots,{\theta_{0:k}^n}^\intercal]^\intercal$. Moreover, at time $k=0$, we eliminate $z_k$ and $k$ from the argument of $\varrho^{\varphi}$ because $z_0$ is given, and we define:
\begin{equation*}
     \varrho_0^{\varphi} (\vect u,\vect w,\vect \theta):=\varrho^{\varphi} (\vect u,\vect w,\vect \theta,z_0,0), 
\end{equation*} 
where $\vect u:=u_{0:N}\in \mathbb U:=\mathcal{U}^N$,  $\vect w:=w_{0:N}\in \mathbb{W}:=\mathcal{W}^N$ and $\vect\theta:=\theta^{1:n}_{0:N}\in \mathbb{T}:=\mathcal{T}_1^N\times\ldots\times \mathcal{T}_n^N$. By defining these notations in the context of STL, the aim will be to provide a control input $\vect u$ such that the function $ \varrho_0^{\varphi}$ becomes positive, or specifically lower bounded by a pre-defined positive robustness level $r_0$. Since this function is affected by the unknown disturbance $\vect w$ and uncertainties $\vect \theta$, the purpose would be to satisfy inequality
\begin{equation*}
    \varrho_0^{\varphi} (\vect u,\vect w,\vect \theta)\geq r_0,
\end{equation*}
 with probability at least a given threshold for all uncertain variables in $\vect \theta$.
 The next section will formulate this problem as a robust chance-constrained program (CCP), and provide two approaches to deal with the probabilistic constraint.

\section{Chance-constrained program}\label{sec:chance}
In this section, we formulate a CCP that satisfies the STL specification with at least a pre-defined probability level (due to the stochasticity of the controlled agent) for all possible behaviors of the other uncontrollable agents, modeled as uncertainty tubes.

Let $(\mathbb{W},\mathcal F, P)$ denote the probability space for the random variable $\vect w\in\mathbb{W}$, where $\mathbb{W}$ is the sample space, $\mathcal F$ is a sigma-algebra on $\mathbb{W}$, and $P:=\mathcal{P}^N$ is the $N$-fold product of the probability measure $\mathcal{P}$. We then define the following CCP:
\begin{subequations}\label{eq:CCP}
\begin{align}
    \min_{\vect u \in \mathbb{U}}\quad & J(\vect u),\\
    \mathrm{s.t.}\quad & P\left\{\varrho_0^{\varphi}(\vect u,\vect w,\vect\theta)\geq r_0\right\}\geq 1-\varepsilon,\quad \forall \vect\theta\in \mathbb{T},   
\end{align}    
\end{subequations}
where $\varepsilon\in(0,1)$ is the constraint violation tolerance and $J$ is a lower semi-continuous cost function. This optimization aims to find an optimal control sequence that minimizes a control-dependent performance function $J$ while ensuring that the STL constraint $\varphi$, with robustness $r_0>0$, is satisfied with probability at least $(1-\varepsilon)$ for all $\vect\theta\in  \mathbb{T}$.

Note that the performance function $J$ is commonly defined as a function of state $x_{0:N}$ and input $\vect u$. However, for the sake of clarity, all state-dependent requirements, such as reaching a specific target set, can be incorporated into the STL specification.

The following remark provides an alternative way of formulating the considered problem as a robust CCP.
\begin{Remark} As an alternative to~\eqref{eq:CCP}, another approach to treat the uncertainty $\vect \theta$ can be formulated as the following CCP:
    \begin{subequations}\label{eq:CCP1}
\begin{align}
    \min_{\vect u \in \mathbb{U}}\quad & J(\vect u),\\
    \mathrm{s.t.}\quad & P\left\{\forall \vect\theta\in \mathbb{T}\, :\, \varrho_0^{\varphi}(\vect u,\vect w,\vect\theta)\geq r_0\right\}\geq 1-\varepsilon.\quad 
\end{align}    
\end{subequations}
It can be shown that this formulation is more conservative than~\eqref{eq:CCP}. To show this, let us assume ${\vect\theta}^\star\in \argmin_{\vect\theta\in \mathbb{T}}P\left\{\varrho_0^{\varphi}(\vect u,\vect w,\vect\theta)\geq r_0\right\}$, then from the optimality, we have:
\begin{equation*}
    \min_{\vect\theta\in \mathbb{T}} \varrho_0^{\varphi}(\vect u,\vect w,\vect\theta)\leq \varrho_0^{\varphi}(\vect u,\vect w,{\vect\theta}^\star),
\end{equation*}
which yields
\begin{equation*}
    {P}\left\{\min_{\vect\theta\in \mathbb{T}} \varrho_0^{\varphi}(\vect u,\vect w,\vect\theta)\!\geq\! r_0\right\} \leq \min_{\vect\theta\in \mathbb{T}}P\left\{\varrho_0^{\varphi}(\vect u,\vect w,\vect\theta)\geq r_0\right\},
\end{equation*}
for all $\vect u\in\mathbb U$. Consequently, the feasible domain of \eqref{eq:CCP1} is a subset of the feasible domain of~\eqref{eq:CCP}, and the constraint satisfaction of~\eqref{eq:CCP1} ensures the satisfaction of~\eqref{eq:CCP}, though the reverse does not hold in general. However, due to the potential infeasibility issues when solving~\eqref{eq:CCP1}, this paper concentrates on solving~\eqref{eq:CCP}.   
\end{Remark}

In the following subsections, we propose two approaches to address the challenges of solving CCP~\eqref{eq:CCP}, arising from trade-offs between conservativeness, computational complexity, and required assumptions. The concentration of measure (CoM) approach is computationally efficient and directly converts the CCP into an expectation-constrained program (ECP) by tightening the probabilistic constraint. The conditional value at risk (CVaR) approach, as a widely used risk measure, on the other hand, uses a minimization within its definition to reformulate the CCP as an ECP.

\subsection{Concentration of measure (CoM)}\label{sec:COM}
In the following, we use the concept of concentration of measure (CoM) to treat the chance constraint in~\eqref{eq:CCP}. This approach was first proposed in~\cite{Soudjani2018}, and here we adopt the approach over dynamical systems with the STL specifications. We make the following assumptions on the distribution of $\vect w$. 

\begin{Assumption}
	\label{ass:light-tailed}
	(Light-tailed distribution) The distribution of random variable $\vect w$ is a Light-tailed distribution. More specifically, there exists $a>1$ such that $$\mathcal C:=\mathbb{E}_P[\exp{\|\vect w\|^a}]<\infty.$$
\end{Assumption}
Assumption~\ref{ass:light-tailed} holds for many distributions, e.g., multivariate normal distribution, exponential distribution, log-normal distribution, and all distributions with bounded support.
\begin{Assumption}
	\label{ass:concentration}
	(CoM property) For stochastic variable $\vect w\in \mathbb W$, there exists a monotonically decreasing function $h:\mathbb R^{\ge 0}\rightarrow[0,1]$ such that
	\begin{equation}
	\label{eq:concentration}
	 P\left\{|f(\vect w)-\mathbb E\left[f(\vect w)\right]|\le t\right\}\ge 1-h(t),\quad\forall t\ge 0,
	\end{equation}
	holds for any Lipschitz continuous function $f:\mathbb{W}\rightarrow\mathbb R$ with Lipschitz constant $1$.
\end{Assumption}

We recall that a function $f:\mathbb{W}\rightarrow \mathbb{R}$ is a Lipschitz continuous function if there exists $L\geq 0$ such that for any two vectors $\vect w_1,\vect w_2\in\mathbb{W}$, $$
    \frac{|f(\vect w_1)-f(\vect w_2)|}{d_{\mathbb{W}}(\vect w_1,\vect w_2)}\leq L<\infty,$$ where $d_{\mathbb{W}}$ denotes a metric on the set $\mathbb W$. Constant $L$ is referred to as the Lipschitz constant. Throughout this paper, we use the 2-norm on $\mathbb{W}$ to calculate the Lipschitz constants, given by $d_{\mathbb{W}}(\vect w_1,\vect w_2)=\sqrt{ (\vect w_1-\vect w_2)^\intercal (\vect w_1-\vect w_2)}$. Note that, by employing alternative metrics, the general results of the paper remain valid, and changing the metric only affects the values of the Lipschitz constant.

Assumption~\ref{ass:concentration} also holds for many distributions. Examples of different distributions with the CoM property and corresponding $h$ functions can be found e.g., in~\cite{Soudjani2018, ledoux2001concentration}. For instance,  the standard multi-variate Gaussian distribution satisfies \eqref{eq:concentration} with $h(t) = \min\{2e^{-2t^2/\pi^2},1\}$~\cite{barvinok1997measure,pisier1999volume}. Observe that if we substitute $h(\cdot)$ with another monotonically decreasing function $\bar h(\cdot)$ such that $\bar h(\cdot)\geq h(\cdot)$, the inequality \eqref{eq:concentration} remains valid when using $\bar h(\cdot)$.

In the following, we make a regularity assumption on the STL predicate functions. This assumption allows us to establish the Lipschitz continuity of robustness functions, formulate a well-defined chance constraint for the robustness function,  and utilize the CoM property.

\begin{Assumption}\label{Assum:predicate} (Lipschitz continuity of predicate functions)
    We assume that the predicate functions are Lipschitz continuous.
\end{Assumption}
Note that, unlike restrictive assumptions in the literature, such as the linearity of predicate functions~\cite{kordabad2024control}, which lead to sets composed of multiple hyperplanes, the assumption above does not impose such restrictions on the shape of the sets. For example, the predicate function $\alpha(z)=1-\|z\|$, which describes the interior of a circle centered at the origin with a radius of one, is Lipschitz continuous with a Lipschitz constant of one.

In the following, we provide two lemmas and a theorem that enable us to use the CoM property in the context of STL and its robustness function.

   \begin{Lemma}\label{lem:ab}
Suppose that $a_1, a_2, b_1, b_2\in\mathbb R$ are such that $a_1\leq a_2$ and $b_1\leq b_2$. Then,   
\begin{equation*}
        |a_1-b_1|\leq \max\{|a_1-b_2|,|a_2-b_1|\}.
\end{equation*}
\end{Lemma} 
\begin{proof}
    Suppose that $a_1\leq b_1$. Then $a_1\leq b_1\leq b_2$ and $|a_1-b_1|\leq |a_1-b_2|$. In the case that  $b_1\leq a_1$, we have $b_1\leq a_1\leq a_2$ and $|a_1-b_1|\leq |a_2-b_1|$.
\end{proof}
\begin{Lemma}\label{lem:lip:minmax}
    For any two Lipschitz functions $f_1:X\rightarrow\mathbb R$ and $f_2:X\rightarrow\mathbb R$, $\max(f_1, f_2)$ and $\min(f_1, f_2)$ are Lipschitz functions with $L := \max(L_1, L_2)$, where $L_i$ is the Lipschitz constant of $f_i$, $i\in\{1,2\}$.
\end{Lemma}
\begin{proof}
 Since the functions $f_1$ and $f_2$ are Lipschitz continuous, we have
    \begin{equation*}
         |f_i(x_1)-f_i(x_2)|\leq L_id_X(x_1,x_2), \,\, i\in\{1,2\}, \,\, \forall x_1,x_2\in X,     
    \end{equation*}
    for some $ L_1, L_2\geq 0$, where $d_X$ is a metric on the set $X$. We first prove the lemma for $\min(f_1, f_2)$. For any $x_1, x_2\in X$, we have
\begin{align}\label{ineq:lem2}
    &|\min(f_1(x_1), f_2(x_1))-\min(f_1(x_2),f_2(x_2))|\leq \max \Big\{
   \nonumber \\ &|f_1(x_1)-f_1(x_2)|, |f_1(x_1)-f_2(x_2)|, |f_2(x_1)-f_1(x_2)|, \nonumber\\  &|f_2(x_1)-f_2(x_2)|\Big\}\leq \max \Big\{|f_1(x_1)-f_1(x_2)|,\nonumber \\ &|f_2(x_1)-f_2(x_2)|\Big\}\leq \max\{L_1,L_2\}d_X(x_1,x_2),
\end{align}
where the first inequality follows from four different cases that might happen for the two $\min$ operators, and the second inequality can be concluded from Lemma~\ref{lem:ab} for $|f_1(x_1)-f_2(x_2)|$ and $|f_2(x_1)-f_1(x_2)|$ cases. For instance, for $|f_1(x_1)-f_2(x_2)|$ case, we can set $a_1=f_1(x_1)$, $a_2=f_2(x_1)$, $b_1=f_1(x_2)$ and $b_2=f_2(x_2)$. From the facts that $\min(f_1(x_1), f_2(x_1))=f_1(x_1)$ and $\min(f_1(x_2), f_2(x_2))=f_2(x_2)$, the inequalities of Lemma~\ref{lem:ab}, stating $a_1\leq a_2$ and $b_1\leq b_2$ are fulfilled and therefore, 
\begin{align*}
 |f_1(x_1)-f_2(x_2)|\leq \max \Big\{&|f_1(x_1)-f_1(x_2)|,\nonumber\\ &|f_2(x_1)-f_2(x_2)|\Big\}.   
\end{align*}
A similar statement holds for the case of $|f_2(x_1)-f_1(x_2)|$ in the second inequality of \eqref{ineq:lem2}. Similarly, we can show that the lemma holds for the $\max$ operator.
\end{proof}

Note that Lemma~\ref{lem:lip:minmax} can be readily extended for the cases where we have more than two functions inside the $\min$ or $\max$ operators and their combinations as stated in the following corollary.

\begin{Corollary}\label{col:lip} Functions $\min(f_1,\ldots,f_n)$ and $\max(f_1,\ldots,f_n)$ are Lipschitz functions with constant $\max\{L_1,\ldots, L_n\}$ with $f_i$ being a Lipschitz function with constant $L_i$ for all $i\in\{1,\ldots,n\}$. Moreover, the result holds for the combination of $\min$ and $\max$ with any number of operators. For instance, $\min(f_1,\ldots,\max(f_j,\ldots,f_{l}),\ldots, f_n)$ with $1\leq j<l<n$, is a Lipschitz function with constant $\max\{L_1,\ldots, L_n\}$.
\end{Corollary}

The following proposition uses Corollary~\ref{col:lip} and shows that the robustness function of an STL specification for linear systems is a Lipschitz function. Moreover, it provides the corresponding Lipschitz constant. 

\begin{Proposition}\label{prop:lip}
\textbf{(Lipschitz constant of the robustness function)}  For any STL specification $\varphi$ with Lipschitz predicate functions satisfying Assumption~\ref{Assum:predicate},  $\varrho_0^{\varphi}(\vect u,\vect w, \vect\theta)$ is Lipschitz continuous with respect to $\vect w$ for the system defined in~\eqref{eq:discsys1} and~\eqref{eq:traj2}. The Lipschitz constant will be $L_{\varphi} = L_1L_2$, where $L_1$ is the maximum Lipschitz constant of predicate functions appearing in $\varphi$. More specifically, consider the STL formula $\varphi$ consists of $\mathcal J$ subformula with predicate functions $\alpha_j$, $j\in \{1,\ldots, \mathcal J\}$, then,
\begin{equation*}
    L_1:= \max_{j\in \{1,\ldots, \mathcal J\}} L_{\alpha_j}, 
\end{equation*}
where $L_{\alpha_j}$ is the Lipschitz constant of $\alpha_j$ for $j\in\{1,\ldots, \mathcal J\}$. Constant $L_2$ is the maximum Lipschitz constant of $z_k$, the state at time $k\in\{1,2,\ldots,N\}$, with respect to $\vect w$, which is bounded by
    \begin{equation}\label{eq:lip:the}
        L_2 = \sqrt{\sum_{i=0}^{N-1}\|\bar A^i
        \|^2},
    \end{equation}
    where $\|\bar A^i
        \|$ is the induced 2-norm of the matrix $\bar A^i$, and $N$ is the length of the sequence $\vect w$.
\end{Proposition}
\begin{proof}
By Assumption~\ref{Assum:predicate}, Corollary~\ref{col:lip} and the system trajectories in~\eqref{eq:dyn:ag}, the mapping $\vect w\rightarrow \varrho_0^{\varphi}$ is Lipschitz continuous for all $\vect u\in \mathbb{U}$ and all $\vect\theta\in \mathbb{T}$.  
To obtain the Lipschitz constant, it can be shown that for the composition of functions, the Lipschitz constant is the multiplication of Lipschitz constants of all functions. 
Therefore, using Corollary~\ref{col:lip} and its extension to multiple $\min$ and $\max$ operators, it is straightforward to show that the robustness function is a Lipschitz function with constant $L_{\varphi}=L_1L_2$, where $L_1$ is the maximum Lipschitz constant of predicate functions appearing in $\varphi$, and $L_2$ is the maximum Lipschitz constant of system trajectory with respect to $\vect w$ over all times in $1\leq k\leq N$. To find the Lipschitz constant of the system trajectory with respect to $\vect w$ for a given time $k$ and for the linear system given in~\eqref{eq:discsys1} and~\eqref{eq:traj2}, we use~\eqref{eq:dyn:ag} with setting $z_0=0$, $\theta^i_j=0$ and $u_j=0$ for all $1\leq i\leq n$ and all $0\leq j\leq k-1$. Then we have $z_{k} = \sum_{j=0}^{k-1}\bar A^{k-j-1} \bar w_j$, where $\bar w_j=[w^\intercal_j, \, \vect 0^\intercal_{n_z-n_x}]^\intercal$ with the following upper bound on the norm:
\begin{align*}
\|z_k\| & =\left\|\sum_{j=0}^{k-1}\bar A^{k-j-1} w_j\right\|\leq \sum_{j=0}^{k-1}\|\bar A^{k-j-1} \bar w_j\|\\
& \leq\!
\sum_{j=0}^{k-1}\!\|\bar A^{k-j-1}\| \|\bar w_j\|\leq\!
\sqrt{\sum_{j=0}^{k-1}\!\|\bar A^{k-j-1}\|^2}
\sqrt{\sum_{j=0}^{k-1}\|\bar w_j\|^2} \\
&\le   \sqrt{\sum_{j=0}^{k-1}\|\bar A^j\|^2} \,\,\|\vect w\|\le L_2 \|\vect w\|,
\end{align*}
where the first and second inequalities use the triangle and Cauchy–Schwarz inequalities, respectively. The obtained bound is valid for all $1\leq k\leq N$, concluding the proof. \end{proof}

In the following, we use Assumption~\ref{ass:concentration} for the Lipschitz continuous function $\varrho_0^{\varphi}(\vect u,\vect w,\vect \theta)$ (cf. Proposition~\ref{prop:lip}) and provide an under-approximation for CCP~\eqref{eq:CCP}.

\begin{Theorem}\label{the:EP} \textbf{(CoM theory)}
Under Assumptions~\ref{ass:concentration} and \ref{Assum:predicate}, the feasible domain of CCP~\eqref{eq:CCP} includes the feasible domain of the following ECP:
\begin{subequations}\label{eq:EP}
\begin{align}
     \min_{\vect u \in \mathbb{U}}\quad & J(\vect u),\\
    \mathrm{s.t.}\quad & \mathbb{E}\left[\varrho_0^{\varphi}(\vect u,\vect w,\vect\theta)\right]-L_\varphi h^{-1} (\varepsilon) \geq r_0,\quad \forall \vect\theta\in \mathbb{T},\label{eq:EP:const}
\end{align}
\end{subequations}
where $h$ is given by the class of distribution according to Assumption~\ref{ass:concentration}, and $L_\varphi$ is the Lipschitz constant of $\varrho_0^\varphi(\vect u,\vect w,\vect\theta)$ with respect to $\vect w$.
\end{Theorem}
\begin{proof}
From Proposition~\ref{prop:lip}, the function $\varrho_0^{\varphi}(\vect u,\vect w, \vect \theta)/L_\varphi$ is Lipschitz continuous with respect to $\vect w$ with Lipschitz constant of one. Under Assumption~\ref{ass:concentration} with $t = h^{-1}(\varepsilon)$, we have
\begin{equation*}\label{eq:concen:phi}	 P\!\left\{\!\left|\frac{\varrho_0^{\varphi}(\vect u,\vect w, \vect \theta)}{L_\varphi}\!-\!\mathbb E_P\!\left[\frac{\varrho_0^{\varphi}(\vect u,\vect w, \vect \theta)}{L_\varphi}\right]\!\right |\!\le\! t\right\}\ge 1-h(t)\!=\!1-\varepsilon,
	\end{equation*}
for all $t\ge 0$, $\vect u\in\mathbb U$ and $\vect\theta\in \mathbb{T}$. Suppose that constraint~\eqref{eq:EP:const} holds, i.e., 
\begin{equation*}
    \mathbb E_{P} \left[\frac{\varrho_0^{\varphi}(\vect u,\vect w, \vect \theta)}{L_\varphi}\right]\geq h^{-1}(\varepsilon)+ \frac{r_0}{L_\varphi}=t+\frac{r_0}{L_\varphi}.
\end{equation*}
Then, with probability at least $(1-\varepsilon)$, we have  
\begin{align*}
\frac{\varrho_0^{\varphi}(\vect u,\vect w, \vect \theta)}{L_\varphi}\ge \mathbb E_P\left[\frac{\varrho_0^{\varphi}(\vect u,\vect w, \vect \theta)}{L_\varphi}\right]-t\ge t+\frac{r_0}{L_\varphi} -t = \frac{r_0}{L_\varphi}.
\end{align*}
This implies that $\varrho_0^{\varphi}(\vect u,\vect w, \vect \theta)\ge r_0$ with probability at least $(1-\varepsilon)$.
\end{proof}

The CCP in~\eqref{eq:CCP} has been reformulated into the ECP in~\eqref{eq:EP} by incorporating constraints on the expectation with a tightening function $L_\varphi h^{-1}$. This reformulation can be addressed by leveraging knowledge of the distribution family or by using an upper bound on the function $h$, related to the CoM property, along with an estimation of the expectation. This approach often simplifies the problem compared to solving the original CCP directly.

In the following subsection, we introduce an alternative method to address chance constraints, utilizing the concept of conditional value at risk (CVaR). CVaR offers an effective way to approximate a given CCP by providing a corresponding ECP.

\subsection{Conditional value at risk (CVaR)}
In order to tackle the chance constraint in CCP~\eqref{eq:CCP}, a natural measure of risk is value-at-risk (VaR). For a random variable $R$ and confidence level $(1-\varepsilon)$, $\mathrm{VaR}_{1-\varepsilon}$ is defined as
\begin{equation*}
       \mathrm{VaR}_{1-\varepsilon}(R):=\inf \{\eta \in \mathbb{R} \,|\,P(R\leq\eta)\geq 1-\varepsilon\},
\end{equation*}
where $P$ is the probability measure of $R$. In fact, $\mathrm{VaR}_{1-\varepsilon}$ represents the worst-case loss with probability at least $(1-\varepsilon)$. Then, one can show that
\begin{equation}\label{eq:VaR}
     \mathrm{VaR}_{1-\varepsilon}(R)\leq 0 \Leftrightarrow  P(R\leq 0 )\geq 1-\varepsilon.
\end{equation}
Unfortunately, $\mathrm{VaR}_{1-\varepsilon}(\cdot)$ is typically a non-convex function, making optimizations with VaR challenging, especially in problems involving high-dimensional or non-normal distributions.

An alternative measure of risk is CVaR. CVaR with a confidence level of $(1-\varepsilon)$, denoted as $\mathrm{CVaR}_{1-\varepsilon}$, measures the expected loss in the $\varepsilon$-tail given that the threshold $\mathrm{VaR}_{1-\varepsilon}$ has been crossed, i.e., 
\begin{equation*}
    \mathrm{CVaR}_{1-\varepsilon}(R):=\mathbb{E}\left [R\, |\, \mathrm{VaR}_{1-\varepsilon}(R)\leq R\right ].
\end{equation*}
In~\cite{rockafellar2000optimization}, the following optimization representation was proposed for CVaR:
\begin{equation}\label{eq:CVaR:opt}
       \mathrm{CVaR}_{1-\varepsilon}(R)=\min_{\eta\in \mathbb{R}}\mathbb{E}\left[\eta+\frac{1}{\varepsilon}(R-\eta)_+ \right],
\end{equation}
where $(\cdot)_+=\max\{0,\cdot\}$. Indeed,  $\mathrm{CVaR}$ is a \textit{coherent} risk measure that satisfies conditions such as convexity and monotonicity~\cite{rockafellar2000optimization}. 
Risk management with CVaR functions can be performed quite efficiently. CVaR can be formulated with convex and linear programming methods, while VaR is comparably more complicated to optimize. Detailed benefits and concepts of CVaR can be found in, e.g., \cite{rockafellar2002conditional}. It can be been shown that~\cite{rockafellar2000optimization}
\begin{equation*}
      \mathrm{VaR}_{1-\varepsilon}(R) \leq \mathrm{CVaR}_{1-\varepsilon}(R).
\end{equation*}
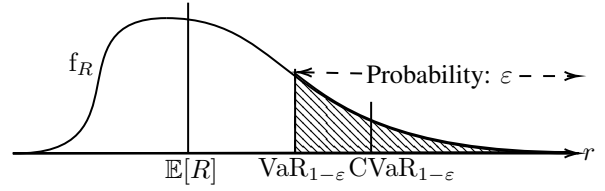
\begin{figure}
\input{tiz1}
\caption{Illustration of the different risk measures $\mathbb E[R]$, $\mathrm{VaR}_{1-\varepsilon}$ and $\mathrm{CVaR}_{1-\varepsilon}$ for random variable $R$ with probability density function $\mathrm{f}_R$.}
\label{fig:CVaR}
\end{figure}
Therefore, the following inequality: 
\begin{equation*}
    \mathrm{CVaR}_{1-\varepsilon}(R)\leq 0,
\end{equation*}
provides an approximation to the chance constraint in~\eqref{eq:VaR}.
Figure~\ref{fig:CVaR} compares the different risk measures including VaR, CVaR and $\mathbb{E}$.

Using the concept of CVaR, CCP~\eqref{eq:CCP} can be under-approximated as follows:
\begin{subequations}\label{eq:CVAR}
\begin{align}
    \min_{\vect u \in \mathbb{U}}\quad & J(\vect u),\\\mathrm{s.t.}\quad & \mathrm{CVaR}_{1-\varepsilon}\left(-\varrho_0^{\varphi}(\vect u,\vect w,\vect\theta)\right)\leq -r_0,\quad \forall \vect\theta\in \mathbb{T},   
\end{align}    
\end{subequations}
Using~\eqref{eq:CVaR:opt}, optimization \eqref{eq:CVAR} can be written as the following ECP:
\begin{subequations}\label{eq:CVAR1}
\begin{align}
    \min_{\vect u \in \mathbb{U},\eta\in\mathbb{R}}\,\, & J(\vect u),\\\mathrm{s.t.}\,\, & 
    \frac{1}{\varepsilon}\mathbb{E}\left[(\varrho_0^{\varphi}(\vect u,\vect w,\vect\theta)-\eta)_- \right]+\eta \geq r_0, \,\, \forall \vect\theta\in \mathbb{T},   
\end{align}    
\end{subequations}
where $(\cdot)_-=\min\{0,\cdot\}$. Using~\eqref{eq:VaR}-\eqref{eq:CVAR1}, it can be seen that any feasible solution of \eqref{eq:CVAR1} is a feasible solution for \eqref{eq:CCP}.

Therefore, both \eqref{eq:EP} and \eqref{eq:CVAR1} provide an approximation for  CCP~\eqref{eq:CCP} based on the CoM and CVaR approaches, respectively. Both approaches lead to an optimization problem that involves an expectation constraint. However, to compute expectations one requires knowledge of the exact distribution $P$. For instance, sample averaging is an often employed technique to approximate expectations numerically. Nonetheless, this approach necessitates a sufficiently large sample size to ensure the accuracy of the empirical expectation compared to the exact one.

In the following, our aim is to solve ECP~\eqref{eq:EP} and \eqref{eq:CVAR1} for an unknown distribution $P$ using finite observed realizations of the disturbance. We then provide finite-sample guarantee on the solution of the resulting data-driven optimization problem.
\section{Distributionally robust program}\label{sec:ecp}
In this section, we provide a distributionally robust optimization (DRO) to solve the ECPs in \eqref{eq:EP} and \eqref{eq:CVAR1} obtained from the CoM and CVaR, respectively.
\subsection{DRO for CoM}
In the following, we aim to solve ECP~\eqref{eq:EP} for an unknown distribution $P$  with respect to the worst-case distribution in an ambiguity set using a Wasserstein distributionally robust approach and provide a finite sample guarantee with respect to the exact ECP~\eqref{eq:EP}. More specifically, the distributionally robust expression of \eqref{eq:EP} can be written as the  following distributionally robust program~(DRP):
\begin{subequations}\label{eq:SP2}
    \begin{align}
        \min_{\vect u\in\mathbb U}\, &J(\vect u),\\
\mathrm{s.t.}\, &\inf_{Q\in\mathbb{Q}}\!\mathbb E_Q\! \left[\varrho_0^{\varphi}(\vect u,\vect w,\vect\theta)\right]\!-\!L_\varphi h^{-1}(\varepsilon)\!\geq \!r_0,\,\forall\vect\theta\!\in \!\mathbb{T},
    \end{align}
\end{subequations}
where $\mathbb{Q}$ is an ambiguity set, defining a set of all distributions around an empirical distribution $\hat Q$ that contains the true distribution $P$ with high confidence. In this paper, we use the Wasserstein metric $W: Q(\mathbb{W})\times  Q(\mathbb{W}) \rightarrow \mathbb{R}_{\geq 0}$ to define the ambiguity ball $\mathbb{Q}$ as
\begin{equation}\label{eq:amb:set}
\mathbb{Q}:=\{Q\in Q(\mathbb{W})\,|\, W(Q,\hat Q)\leq r\},    
\end{equation}
where $Q(\mathbb{W})$ denotes the set of Borel probability measures on the support $\mathbb{W}$ and $r\geq 0$ is the radius of the Wasserstein ball. For any two distributions $Q^1,Q^2\in Q(\mathbb{W})$, the Wasserstein metric $W$ is defined as follows:
\begin{align}\label{eq:WM}
    &W(Q^1,Q^2):= \min_{\kappa\in Q(\mathbb{W}^2)} \Big\{\int_{\mathbb{W}^2} \|  \vect w_1- \vect w_2\|\mathrm{d}\kappa( \vect w_1, \vect w_2)\,\nonumber\\&\qquad\qquad\qquad\qquad\quad\qquad\Big|\, \Pi^j\kappa= Q^j, j=1,2 \Big\},
\end{align}
 where $\Pi^j\kappa$ denotes the $j^{\mathrm{th}}$ marginal of the joint distribution $\kappa$ for $j=1,2$. Note that, the sampling-based reformulation in~\eqref{eq:SP2} stems from the need to make decisions under uncertainty about the true distribution $P$ that governs the random variable $\vect w$. Since the true distribution $P$ is unknown, we rely on a finite number of i.i.d. samples $\{\vect w^i\}_{i=1}^M$ to infer information about $P$. These samples provide an empirical approximation $\hat{Q}$ that can be constructed as follows: 
\begin{equation}\label{eq:hatQ}
    \hat {Q}=\frac{1}{M} \sum_{i=1}^{M} \delta_{\vect w^i},
\end{equation}
where $\delta_{\vect w^i}$ is the Dirac measure concentrated at $\vect w^i$. Consider that, because $\hat{Q}$ is constructed from a limited sample set, it may not perfectly capture the true distribution $P$. To account for the discrepancy between the empirical distribution $\hat Q$ and the true distribution $P$, we introduced the ambiguity set $\mathbb{Q}$ in~\eqref{eq:amb:set}, which includes all distributions that are close to the empirical distribution $\hat{Q}$ within a Wasserstein ball of radius $r$. The parameter $r$ reflects the confidence level: a larger $r$ increases the probability that $\mathbb{Q}$ contains the true distribution $P$, thereby providing a more robust solution to the optimization problem.

Although the DRP~\eqref{eq:SP2} overcomes knowledge of the exact distribution, it is not straightforward to solve since it contains decision variables in the continuous probability measure space. It is desirable to derive a (possibly approximated) solution of~\eqref{eq:SP2} based on the finite samples $\vect w^i$, ensuring both feasibility and performance guarantees.

Inspired by~\cite{mohajerin2018data}, the following theorem offers an equivalent data-driven optimization to the DRP~\eqref{eq:SP2} for the Wasserstein ambiguity set, defined in~\eqref{eq:amb:set}, and with a finite number of samples that eliminates the decision variables in the probability measure space. We can guarantee that any feasible solution obtained from the proposed optimization is a feasible solution to the DRP~\eqref{eq:SP2} and, with a predefined confidence level, is a feasible solution to ECP~\eqref{eq:EP}.

\begin{Theorem}\textbf{(Data-driven DRP for the CoM approach)}\label{thm:1}
Under Assumptions~\ref{ass:light-tailed}-\ref{Assum:predicate} The following optimization:
\begin{subequations}\label{eq:trac0}
\begin{flalign}
&\inf_{\vect u\in\mathbb U, \lambda,\vect y}\,\, J(\vect u),\\
      &\,\,\,\mathrm{s.t.}\,\, \inf_{\vect w\in \mathbb{W},\, \vect \theta\in\mathbb{T}} \left[\varrho_0^\varphi(\vect u,\vect w,\vect\theta)+\lambda \|  \vect w- \vect w^i\|\right]\geq y^i, \nonumber\\&\qquad\,\,\quad\quad\qquad\qquad\qquad\qquad\qquad \forall i\in\{1,\ldots,M\},\label{eq:trac0:cons2} \\&\qquad
 \frac{1}{M} \sum_{i=1}^{M} y^i-\lambda r-L_\varphi h^{-1}(\varepsilon)\geq r_0, \,\,\lambda \geq 0,
\end{flalign} 
\end{subequations}
is equivalent to \eqref{eq:SP2}, with the Wasserstein ambiguity set $\mathbb{Q}$ defined in \eqref{eq:amb:set}-\eqref{eq:hatQ}. Moreover, by choosing the Wasserstein radius $r$, as
\begin{equation}\label{eq:radius}
    r\geq \left \{ \begin{matrix}
\left(\frac{\log(c_1\beta^{-1})}{c_2M}\right)^{\frac{1}{\max\{Nn_x,2\}}} & \mathrm{if}\quad M\geq \frac{\log(c_1\beta^{-1})}{c_2},
\\
\left(\frac{\log(c_1\beta^{-1})}{c_2M}\right)^{\frac{1}{a}} & \mathrm{else},
\end{matrix} \right.
\end{equation}
any feasible solution to~\eqref{eq:trac0} is a feasible solution to~\eqref{eq:EP} with probability at least $(1-\beta)$, for a user-specified confidence level $\beta\in(0, \,1)$, where $c_1, c_2$ are positive constants depending on $a$ and $\mathcal{C}$ (see Assumption~\ref{ass:light-tailed}) and the disturbance dimension~\cite{fournier2015rate}.
\end{Theorem}
\begin{proof}Using the definition of Wasserstein metric in \eqref{eq:WM}, $ \inf_{Q\in\mathbb{Q}}\mathbb E_Q \left[\varrho_0^\varphi(\vect u,\vect w,\vect\theta)\right]$ can be reformulated as 
\begin{align}\label{eq:opt:DR1}
   \inf_{Q\in\mathbb{Q}}\!&\mathbb E_Q\! \left[\varrho_0^\varphi(\vect u,\vect w,\vect\theta)\right]\!\!=\!\!
   \begin{cases}
       \!\inf_{Q,\kappa} \mathbb E_Q \left[\varrho_0^\varphi(\vect u,\vect w,\vect\theta)\right],\\
       \!\mathrm{s.t.} \int_{\mathbb{W}^2} \|  \vect w_1- \vect w_2\|\mathrm{d}\kappa( \vect w_1, \vect w_2)\leq r\nonumber\\\qquad \Pi^1\kappa= Q,\,\, \Pi^2\kappa= \hat Q
   \end{cases}
   \\&=
     \begin{cases}
       \inf_{Q_i} \frac{1}{M} \sum_{i=1}^{M} \mathbb E_{Q_i} \left[\varrho_0^\varphi(\vect u,\vect w,\vect\theta)\right]\\
       \mathrm{s.t.} \frac{1}{M} \sum_{i=1}^{M} \mathbb E_{Q_i} \left[\|  \vect w- \vect w^i\|\right]\leq r.
   \end{cases}
\end{align}
The second equality is a consequence of the law of total probability, which states that any joint distribution $\kappa$ of $\vect w_1$ and $\vect w_2$ can be created from the marginal distribution $\hat Q$ of $\vect w_2$ and the conditional distribution $Q_i$ of $\vect w_1$ given $\vect w_2=\vect w^i$, i.e., $\kappa= \frac{1}{M} \sum_{i=1}^{M} \delta_{ \vect w^i} \otimes Q_i$, where $\otimes$ is used for the product of two probability distributions.
Using the Lagrangian dual problem for the constrained optimization \eqref{eq:opt:DR1}, we obtain
\begin{align}\label{eq:Lag}
    &\inf_{Q\in\mathbb{Q}}\mathbb E_Q \left[\varrho_0^\varphi(\vect u,\vect w,\vect\theta)\right]=\inf_{Q_i}\sup_{\lambda \geq 0} \frac{1}{M} \sum_{i=1}^{M} \mathbb E_{Q_i} \left[\varrho_0^\varphi(\vect u,\vect w,\vect\theta)\right]+\nonumber\\ &\qquad\qquad\lambda \left(\frac{1}{M} \sum_{i=1}^{M} \mathbb E_{Q_i} \left[\|  \vect w- \vect w^i\|\right]-r\right)=\nonumber\\
& \sup_{\lambda\geq 0} \inf_{Q_i}-\lambda r+ \frac{1}{M} \sum_{i=1}^{M} \mathbb E_{Q_i} \left[\varrho_0^\varphi(\vect u,\vect w,\vect\theta)+\lambda \|  \vect w- \vect w^i\|\right]= \nonumber\\&
\sup_{\lambda \geq 0}-\lambda r+\frac{1}{M} \sum_{i=1}^{M} \inf_{\vect w\in \mathbb{W}} \left[\varrho_0^\varphi(\vect u,\vect w,\vect\theta)+\lambda\|  \vect w- \vect w^i\|\right],  
\end{align}
where $\lambda$ is the Lagrange multiplier. Note that the second equality follows from the strong duality that has been shown in \cite{gao2023distributionally}, and the third equality holds because $Q_i$ contains all the Dirac distributions supported on $\mathbb W$, i.e., one can verify that $\sup_{X\in \mathbb{X}} \mathbb{E}^{X}[f(x)]=\sup_x f(x)$, where $\mathbb{X}$ is the set of all distributions $X$ on the support of the random variable $x$.

It is important to note that the strong duality in~\eqref{eq:Lag}, with adjusted notations, has been shown in~\cite{gao2023distributionally} to hold under the following conditions: (1)~semicontinuity of $\varrho_0^\varphi(\vect u,\vect w,\vect\theta)$ and (2) either boundedness of the support set $\mathbb{W}$, or the existence of $\vect w_0\in \mathbb{W}$ such that
\begin{equation}\label{eq:growrate}
   \limsup\limits_{\|\vect w-\vect w_0\|\rightarrow \infty} \frac{\varrho_0^\varphi(\vect u,\vect w_0,\vect\theta)-\varrho_0^\varphi(\vect u,\vect w,\vect\theta)}{\|\vect w-\vect w_0\|}<\infty
\end{equation} 
holds for all $\vect u\in \mathbb U$ and $\vect w\in \mathbb W$. The continuity of the robustness function follows from Proposition~\ref{prop:lip}. Regarding condition~\eqref{eq:growrate}, we have:
\begin{align*}
   &\limsup\limits_{\|\vect w-\vect w_0\|\rightarrow \infty} \frac{\varrho_0^\varphi(\vect u,\vect w_0,\vect\theta)-\varrho_0^\varphi(\vect u,\vect w,\vect\theta)}{\|\vect w-\vect w_0\|}\leq \\ &\limsup\limits_{\|\vect w-\vect w_0\|\rightarrow \infty} \frac{|\varrho_0^\varphi(\vect u,\vect w_0,\vect\theta)-\varrho_0^\varphi(\vect u,\vect w,\vect\theta)|}{\|\vect w-\vect w_0\|}\leq
    \\ &\limsup\limits_{\|\vect w-\vect w_0\|\rightarrow \infty} \frac{L_\varphi \|\vect w-\vect w_0\|}{\|\vect w-\vect w_0\|}=L_\varphi<\infty.
\end{align*} 
Introducing a new auxiliary variable $y^i$, from \eqref{eq:Lag}, we can conclude
\begin{align}\label{eq:supinf1}
&\inf_{Q\in\mathbb{Q}}\mathbb E_Q \left[\varrho_0^\varphi(\vect u,\vect w,\vect\theta)\right]= \\ &
\begin{cases}\nonumber
    \sup_{\lambda,y^i} &-\lambda r+\frac{1}{M} \sum_{i=1}^{M} y^i,\\ \mathrm{s.t.}  &\!\! \inf_{\vect w\in \mathbb{W}}\!\! \left[\varrho_0^\varphi(\vect u,\vect w,\vect\theta)\!+\!\lambda\|  \vect w- \vect w^i\|\right]\!\geq\! y^i, \nonumber\\&\qquad\quad\qquad\qquad\qquad\qquad\qquad \forall i\in\{1,\ldots,M\},\\ &\lambda \geq 0,
\end{cases}
\end{align}
for all $\vect u\in \mathbb U$ and $\vect w\in \mathbb W$. Merging of the decision variables in \eqref{eq:supinf1} with the main optimization in~\eqref{eq:SP2} and the robust optimization for all $\vect\theta\in\mathbb{T}$ in the $\inf$ constraint~\eqref{eq:supinf1} results in~\eqref{eq:trac0}.
Regarding the connection to ECP~\eqref{eq:EP} with the exact distribution, \cite{fournier2015rate} has demonstrated that if the Wasserstein radius $r$ is selected according to~\eqref{eq:radius}, then the probability of containing the exact distribution by the Wasserstein ball $\mathbb Q$ is at least $(1-\beta)$, i.e., $
    {P}^M\left \{P\in \mathbb Q\right \}\geq 1-\beta$. Therefore, for every given $\vect u\in\mathbb U$ and $\vect\theta\in \mathbb{T}$,
    \begin{equation*}
    {P}^M\left \{\inf_{Q\in\mathbb{Q}}\mathbb E_Q \left[\varrho_0^\varphi(\vect u,\vect w,\vect\theta)\right]\leq \mathbb E_P \left[\varrho_0^\varphi(\vect u,\vect w,\vect\theta)\right]\right\}\geq 1-\beta,
\end{equation*}
where ${P}^M$ is the $M$-fold product of the probability measure $P$. This concludes that the feasible domain of \eqref{eq:SP2} (and \eqref{eq:trac0} from the first part of the theorem) is an inner approximation of~\eqref{eq:EP} with a probability of at least $(1-\beta)$.
\end{proof}

Note that the term \textit{data-driven}, as used in e.g.,  \cite{mohajerin2018data, chen2024data, hota2019data}, refers to optimization that utilizes a set of samples to construct the DRP. In the control literature, this term has also been used in contexts where the system matrices are unknown. However, this is beyond the scope of our paper, that assumes the system matrices are known in advance.

Moreover, it is worth mentioning that since $\inf_{\vect w\in \mathbb{W}}$ appears in \eqref{eq:trac0:cons2}, knowing the support $\mathbb{W}$ is necessary to establish the equivalency statement in Theorem~\ref{thm:1}. Alternatively, having access to an over-approximation of the support $\mathbb{W}$ ensures that~\eqref{eq:trac0} is an under-approximation of~\eqref{eq:SP2}.
This means that any feasible solution of~\eqref{eq:trac0} will also be a feasible solution to~\eqref{eq:SP2}. Therefore, it must be assumed that the support $\mathbb{W}$, or an over-approximation of that is available. Note that knowing the support (or an over-approximation of it) is a much milder requirement than knowing the exact distribution, which makes our results applicable to uncertainties with unknown distributions as long as the distribution class of the uncertainty satisfies the raised assumptions.

\subsection{Distributionally robust CVaR}
Analogous to the CoM approach, the distributionally robust expression of \eqref{eq:CVAR} can be written as follows:
\begin{align}\label{eq:CVAR,DRO}
    \min_{\vect u \in \mathbb{U}}\quad & J(\vect u),\\
    \mathrm{s.t.}\quad & \sup_{Q\in\mathbb{Q}}\mathrm{CVaR}^Q_{1-\varepsilon}\left(-\varrho_0^{\varphi}(\vect u,\vect w,\vect\theta)\right)\leq -r_0,\quad \forall \vect\theta\in \mathbb{T},\nonumber   
\end{align}    
where the superscript $Q$ in $\mathrm{CVaR}^Q_{1-\varepsilon}$ indicates the distribution under which the CVaR is calculated. The following theorem presents a data-driven optimization that approximates the feasible set of ~\eqref{eq:CVAR,DRO} and elucidates its connection to ECP~\eqref{eq:CVAR1}.

\begin{Theorem}\textbf{(Data-driven DRP for the CVaR approach)}\label{thm:2}
Under Assumptions~\ref{ass:light-tailed} and \ref{Assum:predicate}, any feasible solution to the following optimization:
\begin{subequations}\label{eq:trac1}
\begin{flalign}
&\inf_{\vect u\in\mathbb U, \lambda,\vect y,\eta}\,\, J(\vect u),\\
      &\,\,\,\mathrm{s.t.}\,\, \inf_{\vect w\in \mathbb{W},\vect\theta\in\mathbb{T}}\! \left[\varrho_0^{\varphi}(\vect u,\vect w,\vect\theta)\!+\!\lambda \|  \vect w- \vect w^i\|\right]\geq y^i\!, \nonumber\\&\qquad\,\,\,\quad\quad\qquad\qquad\qquad\qquad\qquad \forall i\in\{1,\ldots,M\},\label{eq:cons:trac1} \\&\quad\qquad
 0\geq y^i+\eta, \qquad\qquad\qquad\,\,\, \forall i\in\{1,\ldots,M\},
 \\&\quad\qquad
 \frac{1}{M}\! \sum_{i=1}^{M} y^i-\lambda r+(1-\varepsilon)\eta\geq \varepsilon r_0, \quad \lambda\geq 0,
\end{flalign}     
\end{subequations}
is a feasible solution to \eqref{eq:CVAR,DRO}, with the Wasserstein ambiguity set $\mathbb{Q}$ defined in \eqref{eq:amb:set}-\eqref{eq:hatQ}. Moreover, if the Wasserstein radius $r$ is selected according to~\eqref{eq:radius}, any feasible solution to~\eqref{eq:trac1} is a feasible solution to~\eqref{eq:CVAR1} with the probability of at least $(1-\beta)$, for a user-specified confidence level $\beta\in(0, \,1)$.
\end{Theorem}
\begin{proof}
    Using~\eqref{eq:CVaR:opt}, DRP~\eqref{eq:CVAR,DRO} is written equivalently as
\begin{align}\label{eq:ECP,CVAR,DRO}
    \min_{\vect u \in \mathbb{U}}\,\, & J(\vect u),\\
    \mathrm{s.t.}\,\, & \sup_{Q\in\mathbb{Q}}\min_{\eta}\!\mathbb{E}_Q\left[(-\varrho_0^{\varphi}(\vect u,\vect w,\vect\theta)-\eta)_+\right]\!\!+\!\varepsilon\eta\leq\! -\varepsilon r_0, \forall \vect\theta\!\in\! \mathbb{T}.\nonumber
\end{align}    
    Using the minimax inequality in the constraint of~\eqref{eq:ECP,CVAR,DRO}, we have
\begin{align}\label{eq:minmax}
    & \sup_{Q\in\mathbb{Q}}\min_{\eta}\mathbb{E}_Q\left[(-\varrho_0^{\varphi}(\vect u,\vect w,\vect\theta)-\eta)_+\right]+\varepsilon\eta \leq \\ & \min_{\eta} \sup_{Q\in\mathbb{Q}} \mathbb{E}_Q\left[(-\varrho_0^{\varphi}(\vect u,\vect w,\vect\theta)-\eta)_+\right]+\varepsilon\eta,\nonumber
\end{align}
and merging the decision variable $\eta$ into the main optimization of \eqref{eq:ECP,CVAR,DRO}, it can be under-approximated as follows:
\begin{align*}
    \min_{\vect u \in \mathbb{U},\eta}\,\, & J(\vect u),\\
    \mathrm{s.t.}\,\, & \sup_{Q\in\mathbb{Q}}\mathbb{E}^Q\!\left[(-\varrho_0^{\varphi}(\vect u,\vect w,\vect\theta)-\eta)_+\right]+\varepsilon\eta\leq -\varepsilon r_0, \forall \vect\theta\in \mathbb{T},
\end{align*}    
or equivalently
\begin{align*}
    \min_{\vect u \in \mathbb{U},\eta}\quad & J(\vect u),\\
    \mathrm{s.t.}\quad & \inf_{Q\in\mathbb{Q}}\mathbb{E}^Q\left[(\varrho_0^{\varphi}(\vect u,\vect w,\vect\theta)+\eta)_-\right]-\varepsilon\eta\geq \varepsilon r_0,\, \forall \vect\theta\in \mathbb{T},
\end{align*}    
where $(\cdot)_-:=\min\{0,\cdot\}$. Analogous to Theorem~\ref{thm:1}, the following data-driven optimization is obtained:
\begin{subequations}\label{eq:trac2}
\begin{flalign}
&\inf_{\vect u\in\mathbb U, \lambda,\vect y,\eta}\,\, J(\vect u),\\
      &\,\,\,\mathrm{s.t.}\,\, \inf_{\vect w\in \mathbb{W},\, \vect \theta\in\mathbb{T}} \left[(\varrho_0^{\varphi}(\vect u,\vect w,\vect\theta)+\eta)_-+\lambda \|  \vect w- \vect w^i\|\right]\geq y^i,\nonumber\\&\qquad\,\,\,\qquad\qquad\qquad\qquad\qquad\qquad \forall i\in\{1,\ldots, M\},\label{eq:trac1:cons2} \\&\qquad
 \frac{1}{M} \sum_{i=1}^{M} y^i-\lambda r-\varepsilon\eta\geq \varepsilon r_0,\quad \lambda\geq 0.
\end{flalign} 
\end{subequations}
Note that constraint~\eqref{eq:trac1:cons2} can be satisfied by the two following constraints:
\begin{align*}
    &\inf_{\vect w\in \mathbb{W},\vect\theta \in\mathbb{T}} \left[\varrho_0^{\varphi}(\vect u,\vect w,\vect\theta)+\eta+\lambda \|  \vect w- \vect w^i\|\right]\geq y^i,\\
    &\inf_{\vect w\in \mathbb{W}} \left[\lambda \|  \vect w- \vect w^i\|\right]\geq y^i\qquad \Leftrightarrow \qquad 0 \geq y^i.
\end{align*}
Thus, optimization~\eqref{eq:trac1} is obtained by rearranging and changing $y^i-\eta \rightarrow y^i$. The remainder of the proof is derived similarly to that of Theorem~\ref{thm:1}. 
\end{proof}
Theorems~\ref{thm:1}--\ref{thm:2} offer distributionally robust data-driven optimizations using the CoM and CVaR approaches to solve chance constraints. As observed, the use of strong duality in proving Theorem~\ref{thm:1} in the second equality of~\eqref{eq:Lag} makes the data-driven optimization~\eqref{eq:trac0} equivalent to the corresponding DRP~\eqref{eq:SP2}. In contrast, the min-max inequality employed in Theorem~\ref{thm:2} in~\eqref{eq:minmax} leads to the resulting data-driven optimization~\eqref{eq:trac1} being an under-approximation of the corresponding DRP~\eqref{eq:CVAR,DRO}, as strong duality does not generally hold in~\eqref{eq:minmax}. This observation potentially makes the distributionally robust CVaR approach more conservative than the one using the CoM theory. The case study in Section~\ref{sec:case1} will illustrate this observation.

We point out here that, in CCP~\eqref{eq:CCP}, ECPs~\eqref{eq:EP} and ~\eqref{eq:CVAR1}, the constraints must hold for all $\vect\theta\in\mathbb {T}$. While in the corresponding data-driven DRPs in~\eqref{eq:trac0} and \eqref{eq:trac1}, this robustness can be incorporated in the infimum on the disturbance variable and practically does not add more complexity to the optimizations. However, solving the constraints~\eqref{eq:trac0:cons2} and \eqref{eq:cons:trac1} is still challenging due to the existence of $\inf$ in the constraints. This challenge increases upon noticing that, unfortunately, the function $\varrho_0^\varphi$ does not have smoothness properties because it consists of non-smooth $\min$ and $\max$ operators. In the next section, we will detail the smoothing of $\varrho_0^\varphi$ and the numerical considerations in solving the robust optimization.

\section{Solving the formulated robust programs}\label{sec:rp}
In this section, we provide details on the smooth approximation of the robustness function $\varrho_0^\varphi$ and explain how to tackle the $\inf$ in the constraints of \eqref{eq:trac0} and \eqref{eq:trac1} numerically. Moreover, we provide two remarks on potential extensions to the current setting.

The $\min$ and $\max$ operators utilized in defining $\varrho_0^\varphi$ in Section~\ref{sec:stl} are not smooth. Numerical solvers commonly encounter difficulties when these operators appear in the objective function or constraints. Inspired by~\cite{gilpin2020smooth}, we opt for smooth under-approximations for these operators, as
\begin{subequations}
\begin{align*}
\min ([a_1,\ldots,a_m]^\top) & \approx -\frac{1}{C} \log\left(\sum_{i=1}^m \exp({-Ca_i})\right),\\
\max ([a_1,\ldots,a_m]^\top) & \approx \frac{\sum_{i=1}^m a_i \exp({Ca_i})}{\sum_{i=1}^m \exp({Ca_i})},
\end{align*}
\end{subequations}
where $C$ is a positive constant. It is noteworthy that these approximations under-approximate the exact $\min$ and $\max$ operators. Consequently, the robust semantics derived from these approximations are not greater than the original robust semantics written for the STL specification in its positive normal form \cite{kazemi2020formal}, i.e., when the specification is written equivalently in a form that has negation appearing only besides the atomic predicates. Hence, fulfilling the approximated robust semantics ensures the satisfaction of the original semantics directly. Additionally, as demonstrated in~\cite{gilpin2020smooth}, for a sufficiently large $C$, the approximated robust semantics converge to the original semantics with the exact $\min$ and $\max$ operators.
Note that in transitioning from CCP~\eqref{eq:CCP} to ECP~\eqref{eq:EP}, we can assume that the robustness function $\varrho_0^\varphi$ is evaluated using the exact $\min$ and $\max$ operators, ensuring the validity of the Lipschitz constant $L_\varphi$ as obtained in Section~\ref{sec:COM}. We then substitute the expectation of the exact robustness with the expectation of the under-approximated robustness in the constraint~\eqref{eq:EP:const}. Therefore, the results of the paper, and particularly the Lipschitz constant $L_\varphi$, remain valid when employing these under-approximations of the exact $\min$ and $\max$ operators.

Another challenge is how to tackle the $\inf$ in the constraints of \eqref{eq:trac0} and \eqref{eq:trac1} numerically. An efficient numerical method for solving the robust optimization problems in \eqref{eq:trac0} and \eqref{eq:trac1}, which involve inner optimizations within the constraints, is the sequential quadratic programming (SQP) method. SQP is an iterative optimization method used for solving nonlinear constrained optimization problems by solving a series of quadratic subproblems that approximate the original problem at each iteration. At each step of the SQP, it is necessary to compute the values, gradients, and Hessians of both the objective function and the constraints with respect to the primary decision variables. These calculations are straightforward for the objective function and the constraints, typically simple linear or quadratic functions, except for \eqref{eq:trac0:cons2} and \eqref{eq:cons:trac1}, which involve inner optimizations and are more complex to handle. To compute these quantities for the robust constraints, the following steps are used:
\begin{itemize}
    \item In each SQP iteration, the values of the constraints in \eqref{eq:trac0:cons2} and \eqref{eq:cons:trac1} are calculated for a specific decision variable by solving the inner optimizations in the constraints at the current values of that variable.
\item  To compute the gradient and Hessian, a sensitivity analysis of the inner optimizations with respect to the main decision variables is performed, since these variables serve as parameters in the inner optimizations.
\end{itemize}
Therefore, by solving the inner optimizations and performing sensitivity analysis in each SQP step, \eqref{eq:trac0} and \eqref{eq:trac1} can be solved efficiently.

In the following, we present two remarks about possible extensions to the problem setting in the paper.

\begin{Remark}
 The linear consideration of dynamics in Section~\ref{sec:dyn} was applied only in the calculation of the Lipschitz constant of the robustness function in Proposition~\ref{prop:lip}. For linear dynamics, due to the simplicity of the closed-form mapping from disturbance to state trajectory, obtaining the Lipschitz constant is straightforward. However, for nonlinear dynamics, by considering assumptions on the Lipschitz continuity of the nonlinear functions governing the dynamics, it is possible to derive similar theorems by evaluating an upper bound on the Lipschitz constant of the state trajectory with respect to disturbances. Therefore, with appropriate regularity assumptions, the proposed method can also be extended to nonlinear dynamics.
\end{Remark}

\begin{Remark} As mentioned, the objective function in this paper was generally considered as a function of only the control inputs, in the form $J(\vect u)$. However, as utilized in~\cite{kordabad2024distributionally}, this function can also be expressed in a more general form, incorporating inputs, states, and uncertainties. In this case, similar to the feasibility guarantees, with an appropriate selection of the Wasserstein radius, a probabilistic guarantee for the upper bound of the performance function can be established. For more details on this, see~\cite{kordabad2024distributionally} for the single agent setting.
\end{Remark}

\section{Case studies}\label{sec:examples}
In the following, we provide two case studies to illustrate and compare the efficiency of the proposed methods.
\subsection{Comparing CoM and CVaR}\label{sec:case1}
This simple example aims to compare the conservation properties of the CoM approach with the CVaR approach when handling chance constraints and their distributionally robust programs. Consider the following CCP:
\begin{subequations}\label{eq:case1:CCP}
   \begin{align}
\min_{u\in \mathbb{R}}\quad &\mathbb E_{w}[J(u,w)],\label{eq:case1:CCP:cost}\\
    \mathrm{s.t.}\quad  &\mathbb P\{u+w\leq 0\}\geq 1-\varepsilon,\label{eq:case1:CCP:cons}
\end{align} 
\end{subequations}
where $J(u,w):=|u+w|$, $\varepsilon\in(0, 1)$ is the probability threshold, and $w \sim \mathcal{N}(0,1)$ is a random variable, drawn from the standard normal distribution. The chance constraint in~\eqref{eq:case1:CCP:cost} equivalently can be written as
\begin{equation}\label{eq:feas:ccp}
u+F^{-1}_w(1-\varepsilon)= u+{\sqrt{2}\mathrm{erf}^{-1} (1-2\varepsilon)}\leq 0,
\end{equation}
where $F_w$ is the cumulative distribution function of the random variable $w$, and $\mathrm{erf}$ is the error function. Inequality~\eqref{eq:feas:ccp} represents the CCP feasible set for $u$. Note that, the provided CCP is simple enough to evaluate its exact solution in closed form. 
\smallskip
\newline
\textbf{Comparing the resulting DRPs:} In the following, we aim to apply the CoM and CVaR methods to CCP~\eqref{eq:case1:CCP} and compare the resulting ECPs. 
\smallskip
\subsubsection*{CoM}
From the CoM inequality provided in Theorem~\ref{the:EP}, the chance constraint in~\eqref{eq:case1:CCP:cost} can be approximated as follows:
\begin{equation*}
      \mathbb E[u+w]+h^{-1}(\varepsilon)\leq 0 \Rightarrow  u+h^{-1}(\varepsilon)\leq 0,
\end{equation*}
where for the Gaussian distribution, we use a tight $h$ in the form of $h(t)=\min\{e^{-t^2/2},1\}$, and equivalently $h^{-1}(\varepsilon)=\sqrt{2\log(\frac{1}{\varepsilon})}$. Note that the linear function inside the chance constraint is Lipschitz continuous with constant one.
\smallskip
\subsubsection*{CVaR}
Applying the CVaR definition, provided in~\eqref{eq:CVaR:opt} to the chance constraint in~\eqref{eq:case1:CCP} reads as
\begin{equation}\label{eq:case1:cvar}
    \min_{\eta}\left\{\eta+\frac{1}{\varepsilon}\mathbb E\left[(u+w-\eta)_+\right]\right\}\leq 0.
\end{equation}
We then observe that
\begin{align}\label{eq:case1:cvar1}
    \mathbb E[(u+w-\eta)_+]&=\int_{0}^\infty \frac{t}{\sqrt{2\pi}} e^{-\frac{1}{2}\left({t-u+\eta}\right)^2}\mathrm{d}t=\\
    & \frac{1}{\sqrt{2\pi}}\mathrm{e}^{-\frac{1}{2}(-u+\eta)^2}\!+\!\frac{u-\eta}{2}\left(\!\mathrm{erf}(\frac{u-\eta}{\sqrt{2}})\!+\!1\!\right).\nonumber
\end{align}
Substituting~\eqref{eq:case1:cvar1} into~\eqref{eq:case1:cvar}, differentiating inside the optimization in~\eqref{eq:case1:cvar} and setting it to zero results in
\begin{equation*}
    \eta^\star=u+\sqrt{2} \mathrm{erf}^{-1}(1-2\varepsilon),
\end{equation*}
or equivalently,
\begin{equation*}
     \mathrm{erf}\left(\frac{u-\eta^\star}{\sqrt{2}}\right)=2\varepsilon-1,
\end{equation*}   
where $\eta^\star$ is the optimal value of $\eta$. The chance constraint in~\eqref{eq:case1:CCP} can then be approximated using CVaR as follows:
\begin{equation*}
 u+\frac{1}{\varepsilon\sqrt{2\pi}}\mathrm{e}^{-\left(\mathrm{erf}^{-1}\left(1-{2\varepsilon}\right)\right)^2} \leq 0.
\end{equation*}

Figure~\ref{fig:1} compares the feasible domains of the CCP with two approaches for approximating the CCP to an ECP at varying probability thresholds $\varepsilon$. As shown, both methods under-approximate the original CCP, while the CoM results in a smaller feasible domain, making it a more conservative approach for this example.
\begin{figure}[t]
\centering
\includegraphics[width=0.48\textwidth]{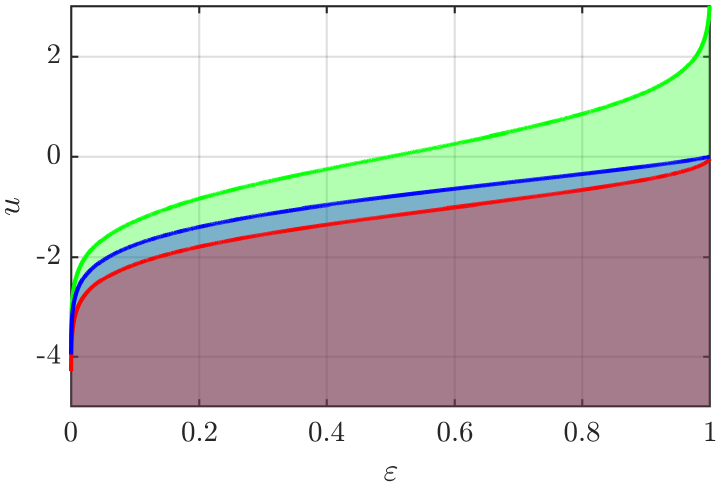}
\caption{Feasible domains for the original CCP (green), the CVaR-based approximation (blue), and the CoM-based approximation (red).}
\label{fig:1}
\end{figure}

\smallskip
\textbf{Comparing the resulting DRPs:}
We then compare distributionally robust formulations of the chance constraint based on the provided approaches.
\smallskip
\subsubsection*{CoM}
For the CoM approach, the distributionally robust expression of the constraint can be written as
\begin{equation}\label{eq:DRO:CoM}
    \sup_{Q\in\mathbb Q} \mathbb E_{Q}[u+w]+h^{-1}(\varepsilon)\leq 0,
\end{equation}
where $\mathbb Q$ is the Wasserstein ambiguity set as defined in~\eqref{eq:amb:set}-\eqref{eq:hatQ}. A data-driven solution to the worst-case expectation appeared in~\eqref{eq:DRO:CoM} can be written as
\begin{align*}
&\sup_{Q\in\mathbb Q} \mathbb E_{Q}[u+w]=\\ &\left\{\begin{array}{ll}
\inf_{y_i,\lambda} &\lambda r+\frac{1}{M}\sum_{i=1}^M y_i\\
    \mathrm{s.t.} &\sup_{w\in\mathbb R}\{u+w-\lambda|w-w_i|\}\leq y_i,\quad 0\leq \lambda,
\end{array}\right.\nonumber\\ &=r+u+\frac{1}{M}\sum_{i=1}^M w_i.
\end{align*}
\smallskip
\subsubsection*{CVaR}
For the distributionally robust CVaR, we have 
\begin{equation*}
    \sup_{Q\in\mathbb Q} \mathrm{CVaR}_{1-\varepsilon}^Q (u+w)\leq 0,
\end{equation*}
and similarly, 
\begin{align*}
   &\sup_{Q\in\mathbb Q} \mathrm{CVaR}_{1-\varepsilon}^Q (u+w)\!\leq\! \inf_{\eta} \eta+\frac{1}{\varepsilon} \sup_{Q\in\mathbb Q}\mathbb E_{Q}[(u+w-\eta)_+]\!\!=\nonumber \\ &\left\{\begin{array}{ll}
\inf_{\eta,y_i,\lambda} &\eta+\frac{1}{\varepsilon}\left\{\lambda r++\frac{1}{M}\sum_{i=1}^M y_i\right\}\\
    \mathrm{s.t.} &\!\!\!\!\!\!\!\!\!\sup_{w\in\mathbb R}\{(u+w-\eta)_+\!-\!\lambda|w-w_i|\}\leq y_i, 0\leq \lambda,
\end{array}\right.\nonumber\\&= \frac{r}{\varepsilon}+\inf_{\eta}\left\{\eta+\frac{1}{M \varepsilon }\sum_{i=1}^M \max(0,w_i+u+\eta)\right\}.
\end{align*}
We now compare these two approaches. We consider samples drawn from the distribution $\mathcal{N}(0,1)$ and let $M\rightarrow \infty$. One can see that in the DRP-CoM an addition term with respect to the ECP is $r$, while $r/\varepsilon$ is the additional term for distributionally robust CVaR. Figure~\ref{fig:2} (top) compares the feasible domains of DRP-CoM and DRP-CVaR for a given Wasserstein radius $r=1$ and varying probability threshold $\varepsilon$. This figure shows that DRP-CoM is less conservative than DRP-CVaR for $\varepsilon< 0.748$. Figure~\ref{fig:2} (bottom) compares these two approaches for a fixed $\varepsilon=0.6$ and varying Wasserstein radius. This figure also shows that DRP-CoM is less conservative than DRP-CVaR for $r> 0.56$.

\begin{figure}[t]
\centering
\includegraphics[width=0.48\textwidth]{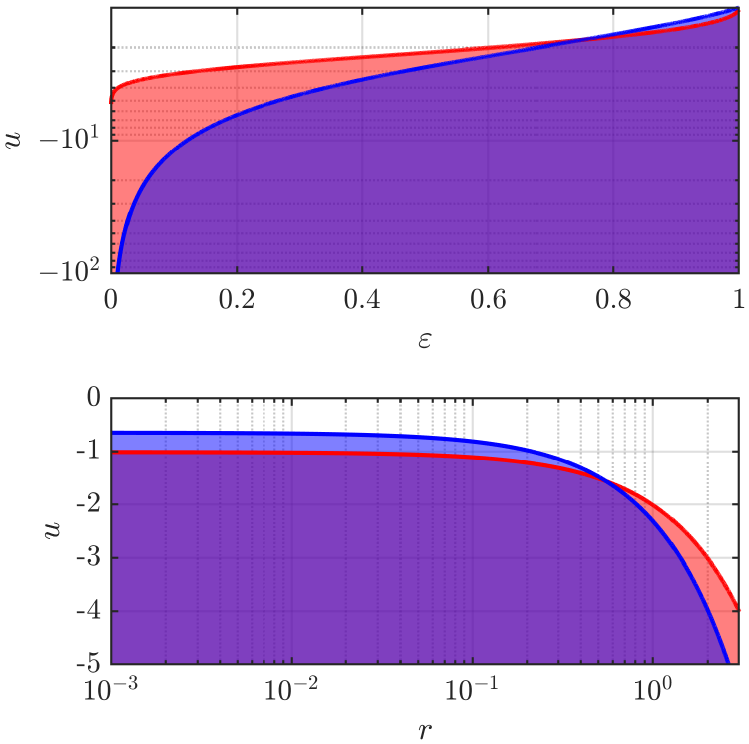}
\caption{Feasible domain for the DRP-CoM (in red) and DRP-CVaR (in blue) optimizations for varying probabilistic thresholds $\varepsilon$ and fixed Wasserstein radius $r=1$ (top), and for varying Wasserstein radii $r$ and fixed probabilistic threshold $\varepsilon=0.6$ (bottom). It can be observed that, depending on the problem parameters such as $\varepsilon$ and $r$, either the CVaR or the CoM approach may be employed effectively.}
\label{fig:2}
\end{figure}

Figure~\ref{fig:3} compares the distributions of objective functions and their expectations for these two approaches, as well as their DRP formulation, for $\varepsilon=0.6$ and $r=1$. The top (bottom) figure compares the CoM (CVaR) ECP solution, denoted by $u^\star_{\mathrm{CoM}}$ ($u^\star_{\mathrm{CVaR}}$) in red, with its DRP solution $\hat u_{\mathrm{CoM}}$ ($\hat u_{\mathrm{CVaR}}$) in blue. As can be observed, the blue distribution exhibits larger values, highlighting the robustness of the DRP approach.  

\begin{figure}[t]
\centering
\includegraphics[width=0.48\textwidth]{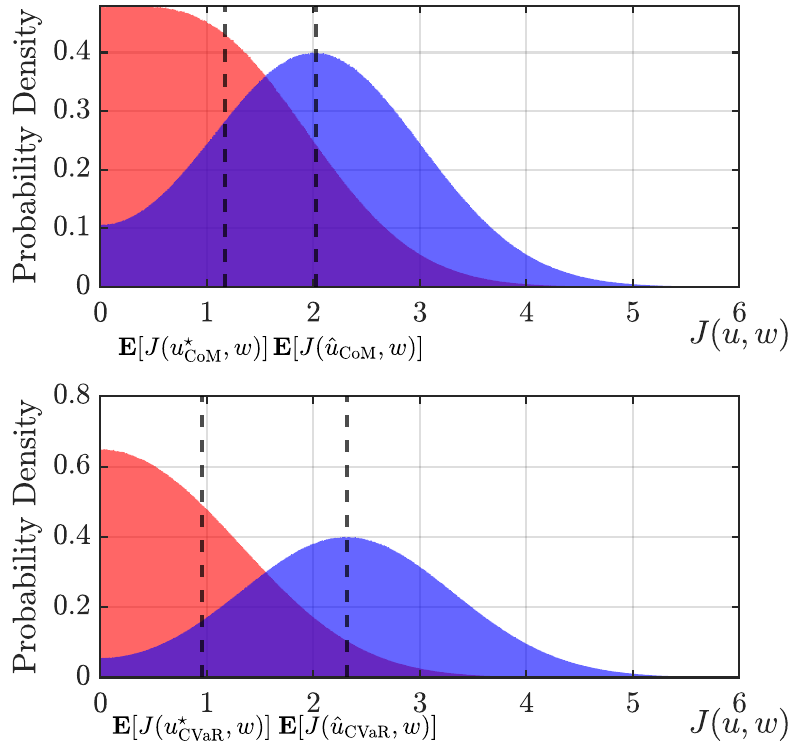}
\caption{Comparison of the objective function distribution for CoM (top) and CVaR (bottom) with their respective DRP solutions.}
\label{fig:3}
\end{figure}

\subsection{Two-agents scenario}
In this example, we address control synthesis for a stochastic controlled agent interacting with an uncontrollable agent. The objective is to minimize control effort while ensuring compliance with a probabilistic STL specification. The STL specification for the controlled agent involves reaching a specific target set while consistently avoiding collisions with the uncontrollable agent, accounting for uncertainties in the behavior of the uncontrollable agent.

We consider the controlled agent to have the following dynamics:
\begin{equation*}
x_{k+1}\!=\!\!\begin{bmatrix}
1 & 0\\ 
0 & 0.8
\end{bmatrix} x_k+ \begin{bmatrix}
0 & -0.1\\ 
0 & 0
\end{bmatrix}\!y_k+ \begin{bmatrix}
1 & 0\\ 
0 & 1
\end{bmatrix}\!u_k+0.005\! \begin{bmatrix}
w_1 \\ 
w_2
\end{bmatrix},    
\end{equation*}
where $w_i\sim \mathcal{N}(0,1)$ (assumed to be unknown), and the uncontrollable agent is modeled as follows,
\begin{equation*}
        y_{k+1}=\begin{bmatrix}
0.9 & 0\\ 
0 & 1
\end{bmatrix} y_k+ \begin{bmatrix}
0.2 & 0\\ 
0 & -0.1
\end{bmatrix}x_k+ 0.005\begin{bmatrix}
\theta_1 \\ 
\theta_2
\end{bmatrix},
\end{equation*}
where $ \|[\theta_1, \theta_2]\|\leq 1$. We aim to design a control strategy for the controlled agent to reach the region $\|x-[0.5,\,0.5]^\top \|\leq 0.1$ within $20$ steps, while consistently avoiding the other agent for all $ \|[\theta_1, \theta_2]\|\leq 1$ with at least $0.9$ probability, and minimizing a quadratic control effort, given by $J(\vect u)=\sum_{0}^{19} u^\intercal_ku_k$. By setting $x_0=[0,\,0]$ and $y_0=[0.2,\,0.3]^\top$, Figure~\ref{g:1} shows the trajectories of both agents, obtained by applying the sample-average approximation method for the corresponding ECPs extracted from the CoM and CVaR approaches. Both approaches result in fairly similar paths, successfully fulfilling the STL task.

\begin{figure}[t]
\centering
\includegraphics[width=0.48\textwidth]{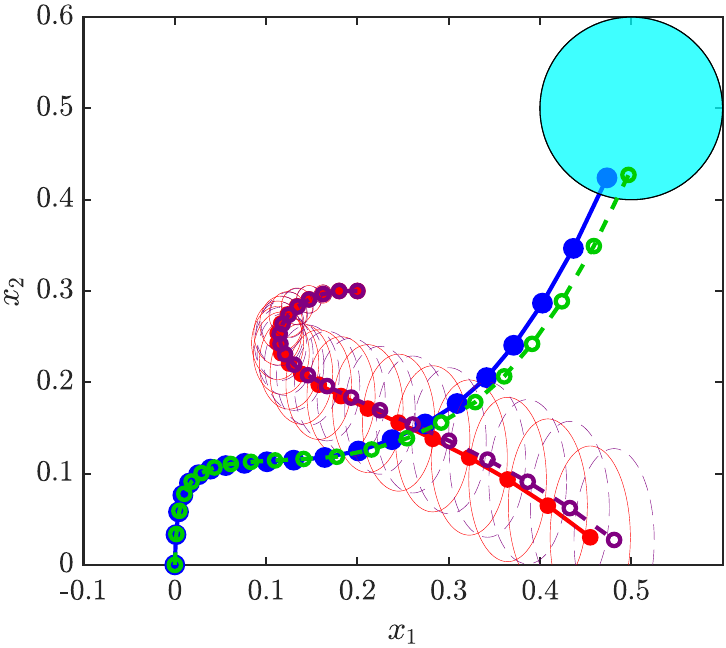}
\caption{Trajectories of the controlled agent are shown in solid blue (dashed green), and the corresponding uncontrollable agent path in solid red (dashed purple), with its possible uncertainties represented by ellipsoids, using the CoM (CVaR) method and sample-average approximation.}
\label{g:1}
\end{figure}

Figure~\ref{g:5} compares the DRP-CoM and DRP-CVaR approaches. It can be seen that in the distributionally robust version, CVaR is more conservative in both avoidance and reachability tasks.

\begin{figure}[t]
\centering
\includegraphics[width=0.48\textwidth]{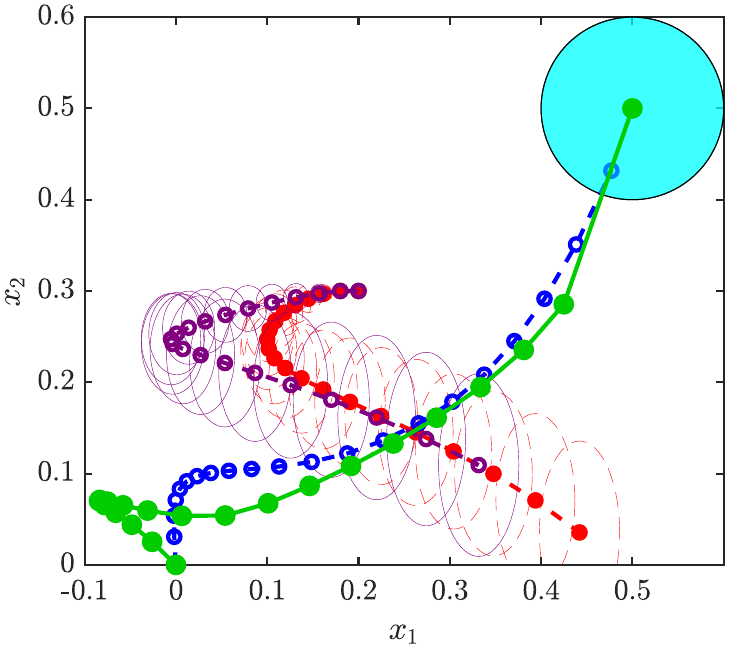}
\caption{Trajectories of the controlled agent are shown in dashed blue (solid green), and the corresponding uncontrollable agent path in dashed red (solid purple), with its possible uncertainties represented by ellipsoids, using the DRP-CoM (DRP-CVaR) method.}
\label{g:5}
\end{figure}

Figure~\ref{g:4} compares the minimum distance between the controlled agent and the uncertain ellipsoid, representing all possible uncontrollable agent locations. In both the sample-average method and DRP, the CVaR approach yields a more conservative solution (i.e., greater distance), shown in green, compared to the CoM approach, shown in blue. The dashed lines correspond to the DRP solution with the same color. As expected, the DRP solution is more conservative in both approaches compared to their ECP solutions.
\begin{figure}[t]
\centering
\includegraphics[width=0.48\textwidth]{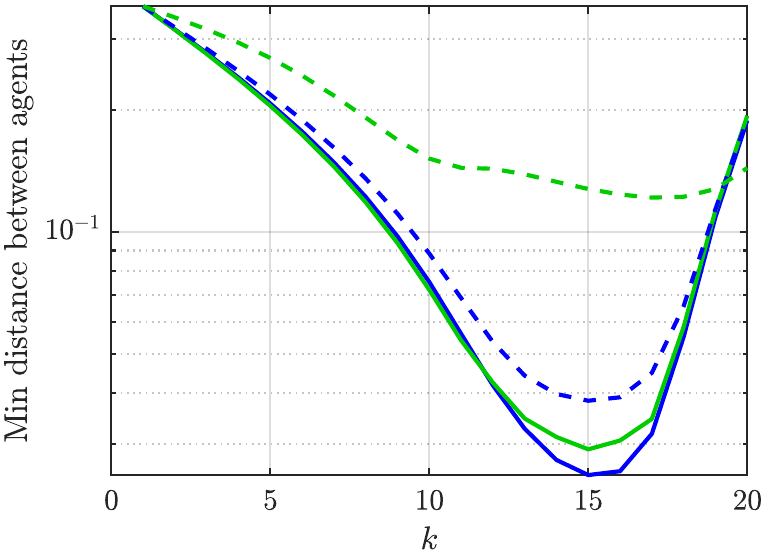}
\caption{The minimum distance between the agents is shown for the CoM approach (in blue) and the CVaR approach (in green), with their corresponding DRP solutions represented by dashed lines in the same colors.}
\label{g:4}
\end{figure}

Figure~\ref{g:2} illustrates 200 different trajectories for both agents. The top two subfigures show the sample-average method, with the left representing the CoM method and the right representing the CVaR method. The bottom two subfigures display the corresponding DRP solutions for each approach. Among these trajectories, $20$ ($14$) failed to meet the STL specification for the ECP solution using the CoM (CVaR) method, while $14$ ($12$) failed for the DRP solution using CoM (CVaR).

\begin{figure}[t]
\centering
\includegraphics[width=0.48\textwidth]{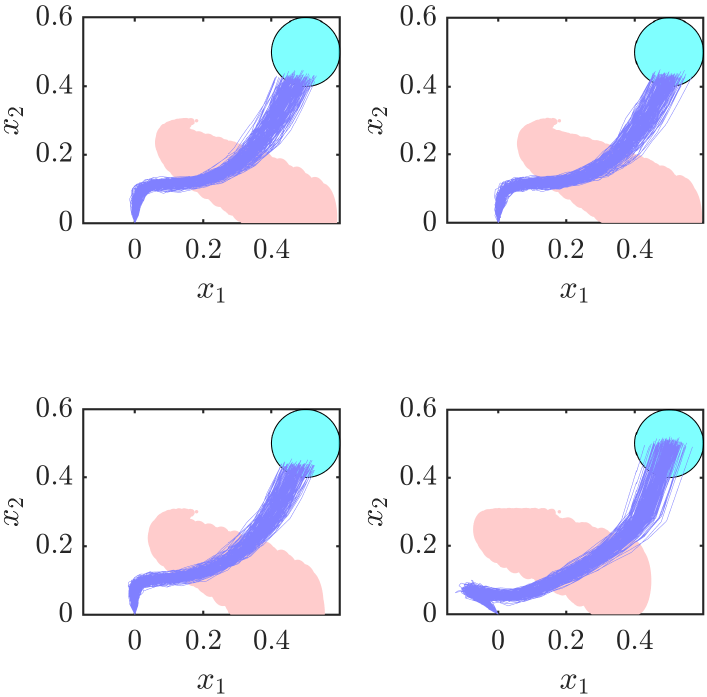}
\caption{Different trajectories for both agents using the CoM method (top left) and CVaR method (top right) in the sample-average approach for ECP, along with their corresponding DRP solutions (bottom left and bottom right).}
\label{g:2}
\end{figure}

Figure~\ref{g:3} displays the distribution of the robustness function and its expectations using the CoM method (red) and the CVaR method (blue) for both the sample-average ECP (top) and DRP (bottom) solutions. It is important to note that the positivity of the robustness measure indicates satisfaction of the STL specification. These figures further confirm the conservativeness of the CVaR approaches in both cases and demonstrate how a distributionally robust solution can yield a more positive robustness function with high probability.

\begin{figure}[t]
\centering
\includegraphics[width=0.48\textwidth]{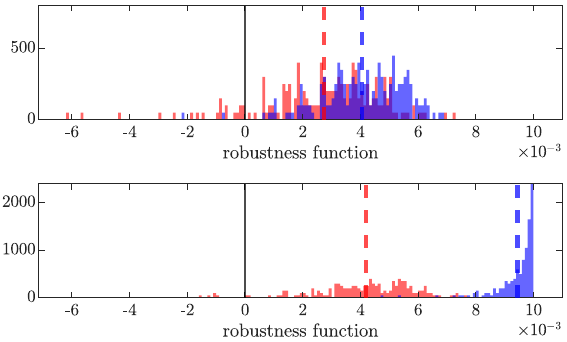}
\caption{The distribution of the robustness function and its expectations using the CoM method (red) and the CVaR method (blue) for both the sample-average ECP (top) and DRP (bottom) solutions.}
\label{g:3}
\end{figure}

Figure~\ref{g:6} illustrates the control inputs for both the CoM method (blue) and the CVaR method (dashed red) using the sample-average ECP (top) and DRP (bottom). The figure shows that in the ECP, both methods yield almost similar control inputs, whereas in the DRP, the CVaR approach produces greater variation in the control inputs. The comparison of the objective cost is presented in Figure~\ref{g:7}, highlighting the effect of the conservativeness of the CVaR approach on the cost function in both the sample-average ECP (two bars on the left) and DRP (two bars on the right) solutions.

\begin{figure}[t]
\centering
\includegraphics[width=0.48\textwidth]{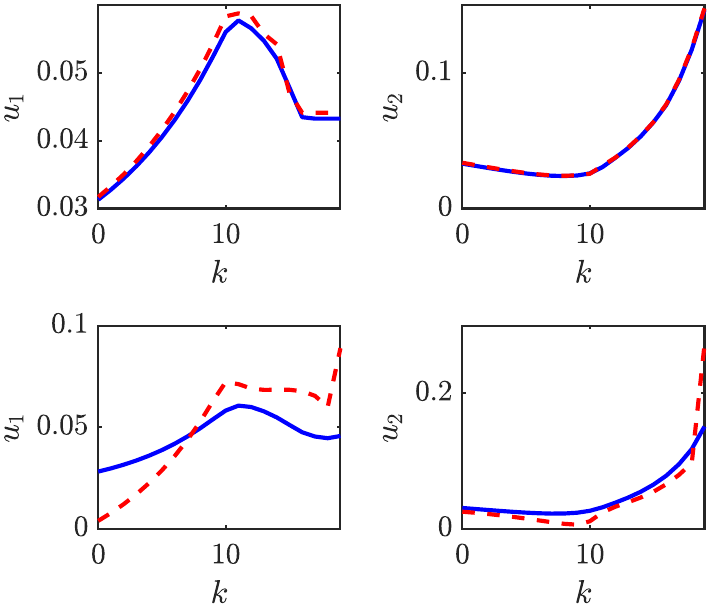}
\caption{Control inputs for the CoM method (blue) and the CVaR method (dashed red) using the sample-average ECP (top) and DRP (bottom).}
\label{g:6}
\end{figure}

\begin{figure}[t]
\centering
\includegraphics[width=0.48\textwidth]{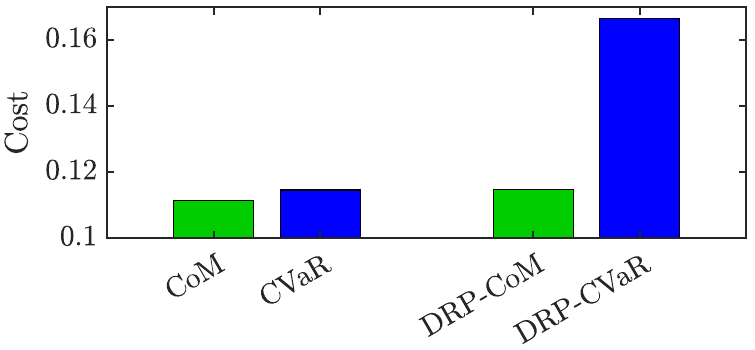}
\caption{Comparison of the objective cost associated with the CoM (CoM) and CVaR methods in green and blue, respectively, for both sample-average ECP (two bars on the left) and DRP (two bars on the right) solutions.}
\label{g:7}
\end{figure}

\section{Conclusions}\label{sec:concl}
This paper addressed chance-constrained programs (CCP) for systems with stochastic dynamics and signal temporal logic (STL) specifications, aiming to ensure that a controlled agent interacting with uncertain uncontrollable agents satisfies probabilistic STL requirements. The two proposed approaches, based on concentration of measure (CoM) and conditional value at risk (CVaR), transformed CCPs into expectation-constrained programs (ECPs), simplifying the optimization process. By employing a Wasserstein distributionally robust program (DRP), the paper introduced a data-driven method that provided probabilistic guarantees on the feasibility and performance of the control strategy. Numerical case studies confirmed the effectiveness of both methods in ensuring STL satisfaction across various problem settings.

For future work, we plan to extend the framework to distributed control for multi-agent systems to improve scalability and reduce computational complexity for applications such as large-scale autonomous driving and robotics. In addition, investigating possible methods based on scenario optimization to solve CCP with the STL robustness function, which is generally non-convex, can be considered for further research.

\bibliographystyle{IEEEtran}
\bibliography{DRSTL}  
\begin{IEEEbiography}[{\includegraphics[width=1in,height=1.25in,clip,keepaspectratio]{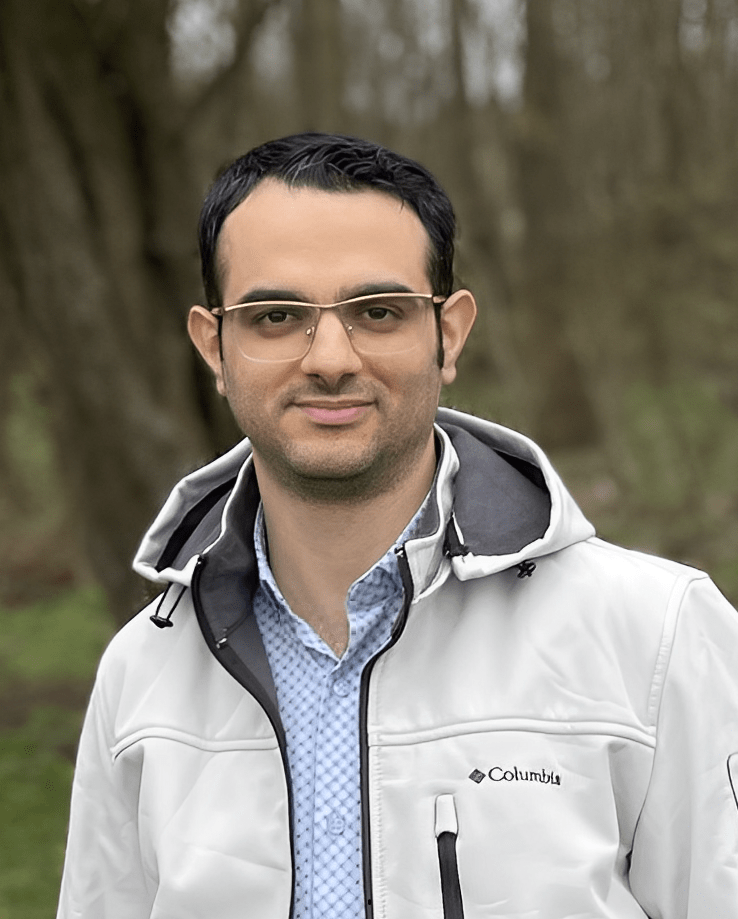}}]{Arash Bahari Kordabad} received his Ph.D. degree from the Department of Cybernetics Engineering at the Norwegian University of Science and Technology (NTNU), Trondheim, Norway, in 2023. Before joining NTNU, he got his B.Sc. in mechanical engineering at the University of Tabriz, Tabriz, Iran, in 2017, and his M.Sc. degree in mechanical engineering from the Sharif University of Technology, Tehran, in 2019. He currently works at the Max Planck Institute for Software Systems, Kaiserslautern, Germany, as a postdoctoral researcher. His research interests include safe reinforcement learning, model predictive control, and optimization for autonomous vehicles, smart-grid applications, and multi-agent systems.
\end{IEEEbiography}

\vspace{-0.9cm}
\begin{IEEEbiography}[{\includegraphics[width=1in,height=1.25in,clip,keepaspectratio]{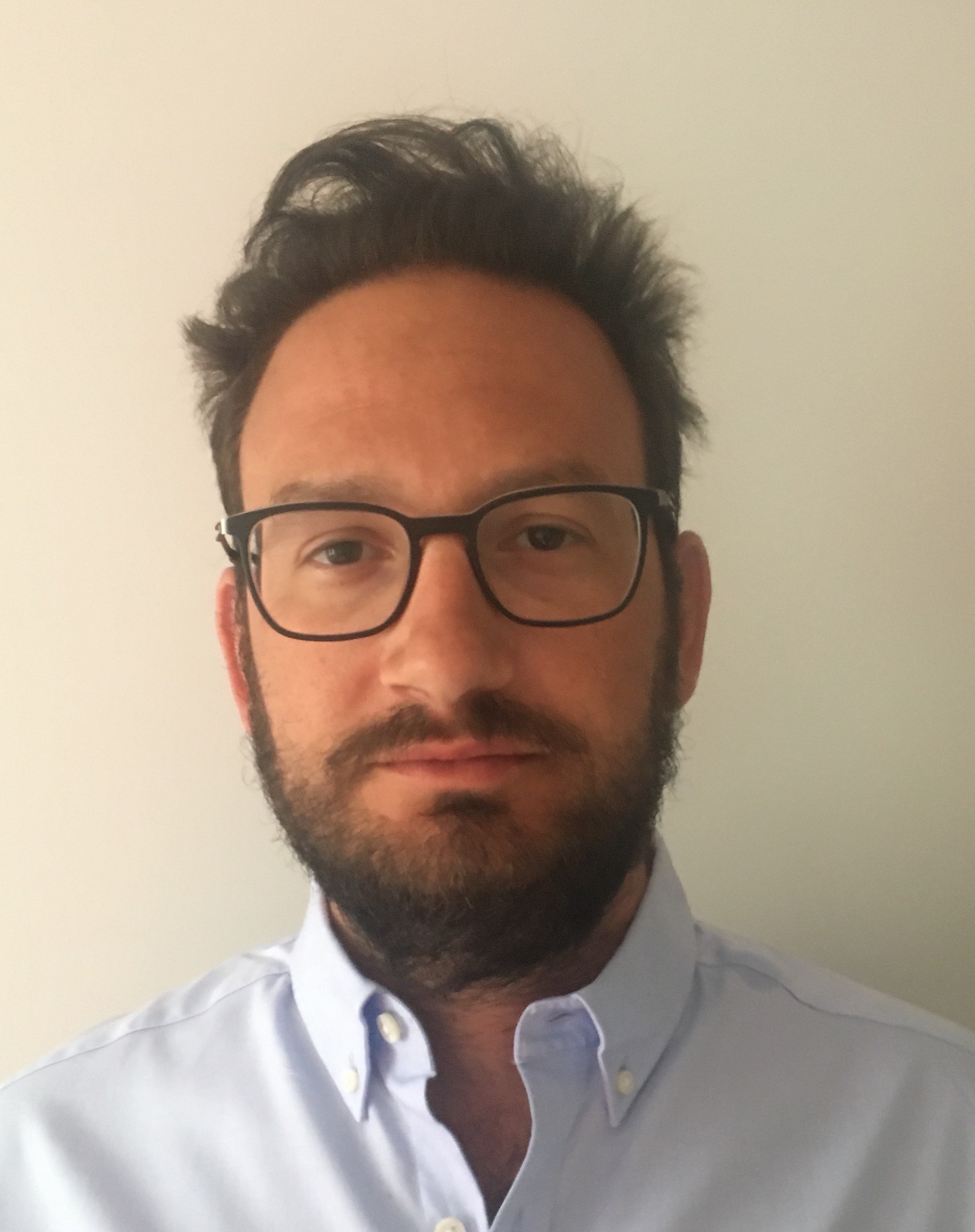}}]{Eleftherios E. Vlahakis} received the M.Eng. degree in electrical and computer engineering from the Aristotle University of Thessaloniki, Thessaloniki, Greece, in 2007, the M.Sc. degree in energy management and production from the National Technical University of Athens, Athens, Greece, in 2014, and the Ph.D. degree in Control Theory from the City University of London, U.K., in 2020. From 2009 to 2015, he worked in various industrial roles in large-scale power networks and renewable energy systems. From 2020 to 2022, he was a Research Fellow with the School of EEECS, Queen’s University Belfast, U.K. He is currently a Postdoctoral Researcher with the Division of Decision and Control Systems, KTH Royal Institute of Technology, Stockholm, Sweden. His 
research interests include distributed control, multi-agent systems, and hybrid systems.
\end{IEEEbiography}

\begin{IEEEbiography}[{\includegraphics[width=1in,height=1.25in,clip,keepaspectratio]{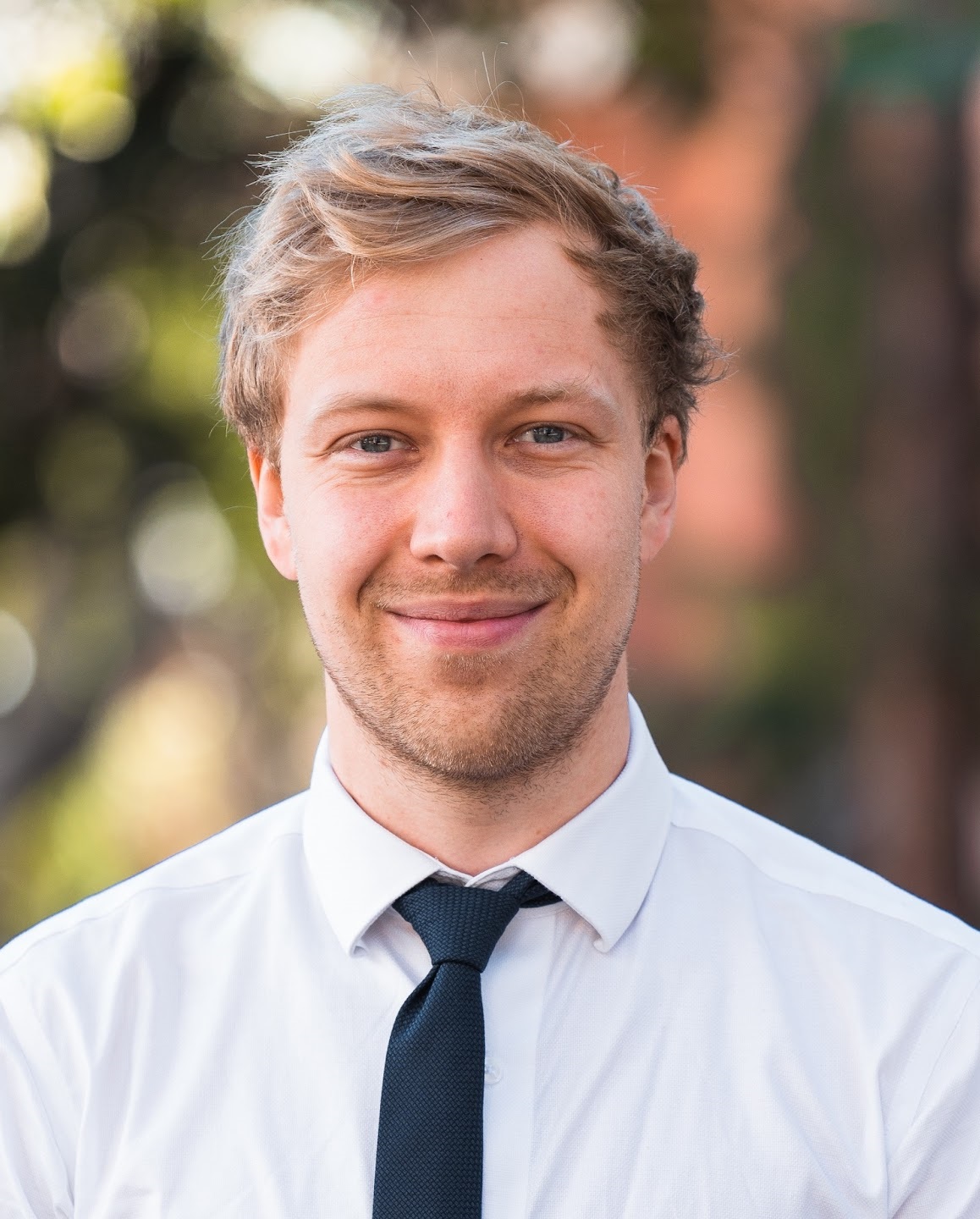}}]{Lars Lindemann} received the Ph.D. degree in electrical engineering from the KTH Royal Institute of Technology, Stockholm, Sweden, in 2020. He is currently an Assistant Professor with the Thomas Lord Department of Computer Science, University of Southern California, Los Angeles, CA, USA, where he is also a Member of the Ming Hsieh Department of Electrical and Computer Engineering (by courtesy), Robotics and Autonomous Systems Center, and Center for Autonomy and Artificial Intelligence. Between 2020 and 2022, he was a Postdoctoral Fellow with the Department of Electrical and Systems Engineering at the University of Pennsylvania. His research interests include systems and control theory, formal methods, and autonomous systems. He was the recipient of the Outstanding Student Paper Award at the 58th IEEE Conference on Decision and Control and the Student Best Paper Award (as a co-advisor) at the 60th IEEE Conference on Decision and Control. He was finalist for the Best Paper Award at the 2022 Conference on Hybrid Systems: Computation and Control and for the Best Student Paper Award at the 2018 American Control Conference.
\end{IEEEbiography}
\vspace{-1cm}
\begin{IEEEbiography}[{\includegraphics[width=1in,height=1.25in,clip,keepaspectratio]{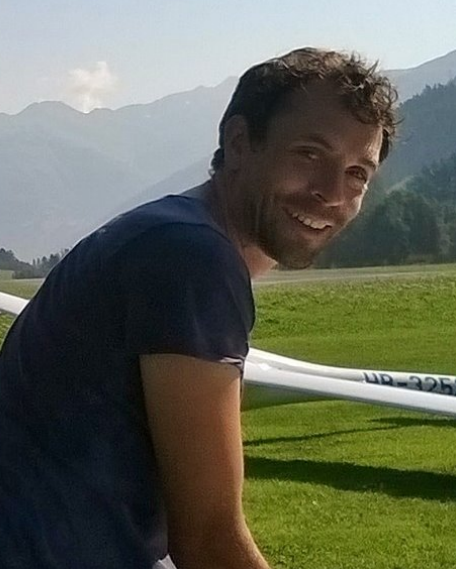}}]{Sebastien Gros} received his Ph.D degree from EPFL, Switzerland, in 2007. After a journey by bicycle from Switzerland to the Everest base camp in full autonomy, he joined an R\&D group hosted at Strathclyde University focusing on wind turbine control. In 2011, he joined the university of KU Leuven, where his main research focus was on optimal control and fast NMPC for complex mechanical systems. He joined the Department of Signals and Systems at Chalmers University of Technology, G\"{o}teborg in 2013, where he became associate Prof. in 2017. He is now full Prof. at NTNU, Norway, head of department at department of engineering cybernetics and affiliate Prof. at Chalmers. His main research interests include numerical methods, real-time optimal control, reinforcement learning, and the optimal control of energy-related applications.
\end{IEEEbiography}
\vspace{-1cm}
\begin{IEEEbiography}[{\includegraphics[width=1in,height=1.25in,clip,keepaspectratio]{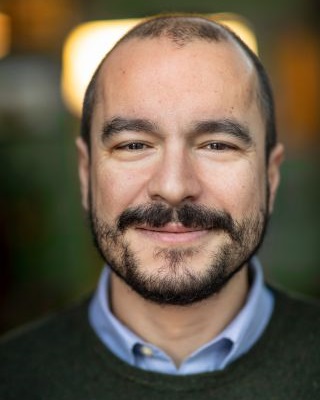}}]{Dimos V. Dimarogonas} was born in Athens, Greece, in 1978. He received the Diploma in Electrical and Computer Engineering in 2001 and the Ph.D. in Mechanical Engineering in 2007, both from the National Technical University of Athens (NTUA), Greece. Between 2007 and 2010, he held postdoctoral positions at the KTH Royal Institute, Stockholm Sweden, and the Laboratory for Information and Decision Systems (LIDS), MIT, Boston USA. He is currently a Professor and Head at the Division of Decision and Control, KTH Royal Institute of Technology, Stockholm, Sweden. His current research interests include Multi-Agent Systems, Hybrid Systems and Control, Robot Navigation and Networked Control. He serves in the Editorial Board of Automatica, and the IEEE Transactions on Control of Network Systems. He is an IEEE Fellow.
\end{IEEEbiography}

\vspace{-1cm}
\begin{IEEEbiography}[{\includegraphics[width=1in,height=1.25in,clip,keepaspectratio]{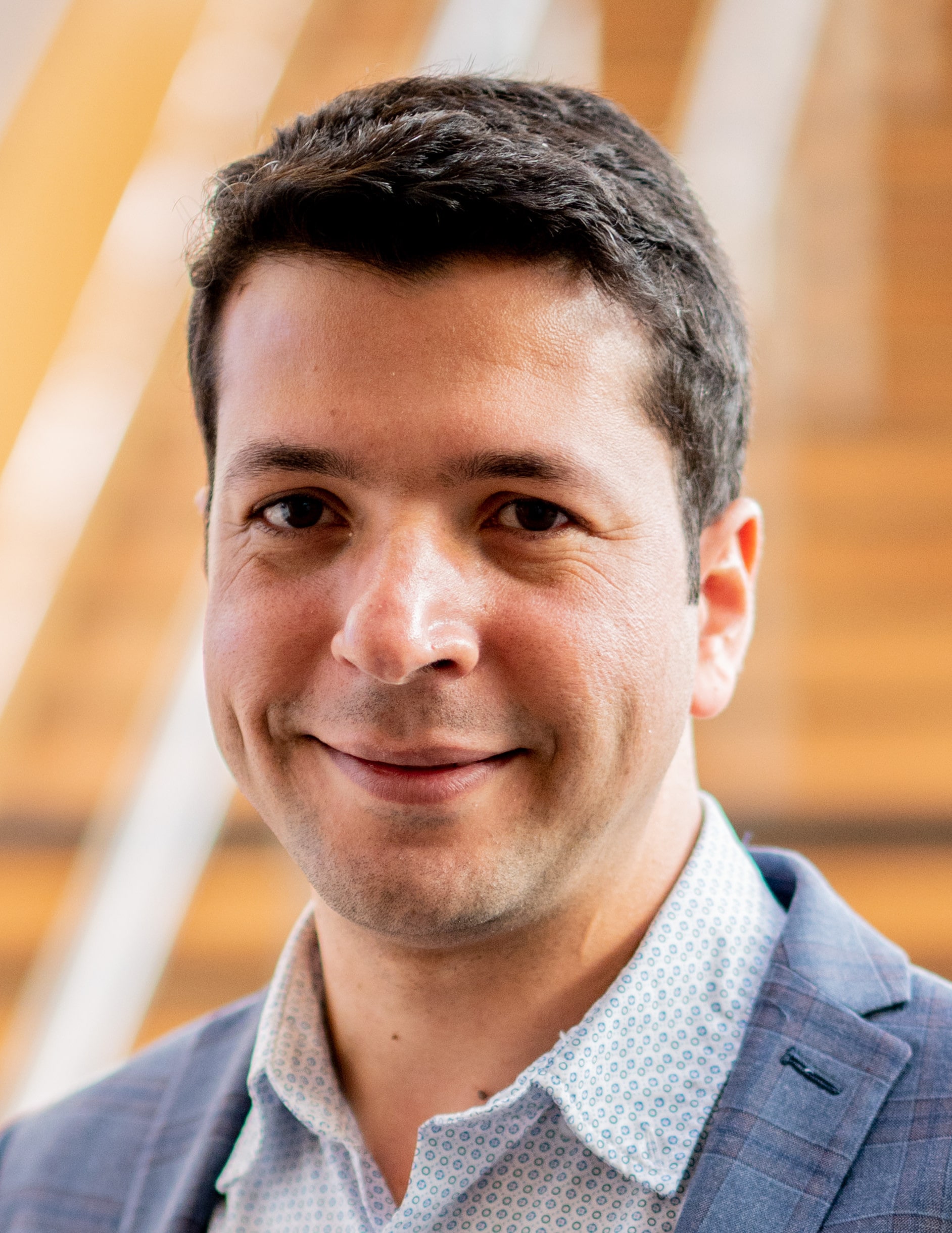}}]{Sadegh Soudjani} is a Research Group Leader at the Max Planck Institute for Software Systems, Germany. Previously, he was the Director of the AMBER Group at Newcastle University, United Kingdom, and Professor in Cyber-Physical Systems at Newcastle University. He received the B.Sc. degrees in mathematics and electrical engineering, and the M.Sc. degree in control engineering from the University of Tehran, Tehran, Iran, in 2007 and 2009, respectively. He received the Ph.D. degree in systems and control in November 2014 from the Delft Center for Systems and Control at the Delft University of Technology, Delft, the Netherlands. Before joining Newcastle University, he was a Postdoctoral Researcher at the Department of Computer Science, University of Oxford, United Kingdom, and at the Max Planck Institute for Software Systems, Germany.
\end{IEEEbiography}

\end{document}

%% file: tiz1.tex
\tikzset{every picture/.style={line width=0.75pt}}


\begin{tikzpicture}[x=0.75pt,y=0.75pt,yscale=-0.65,xscale=0.65]

\draw  [dash pattern={on 4.5pt off 4.5pt}]  (341,211) -- (341,140) ;
\draw [pattern=north west lines, pattern color=black] (341,211) -- (341,150)--(341,150) .. controls (372,174) and (419,211) .. (555,211)--(122,211.99) -- (559,211.99);

\draw    (122,211.99) -- (559,211.99) ;
\draw [shift={(561,212)}, rotate = 180] [color={rgb, 255:red, 0; green, 0; blue, 0 }  ][line width=0.75]    (10.93,-3.29) .. controls (6.95,-1.4) and (3.31,-0.3) .. (0,0) .. controls (3.31,0.3) and (6.95,1.4) .. (10.93,3.29)   ;
\draw    (241,107) .. controls (345,107) and (326,211.99) .. (555,211.99) ;
\draw    (122,211.99) .. controls (228,214) and (151,107) .. (241,107) ;
\draw   (258,211.99) -- (258,95) ;
\draw    (400,211.99) -- (400,172) ;
\draw  [dash pattern={on 4.5pt off 4.5pt}]  (393,150) -- (345,149.04) ;
\draw [shift={(343,149)}, rotate = 1.15] [color={rgb, 255:red, 0; green, 0; blue, 0 }  ][line width=0.75]    (10.93,-3.29) .. controls (6.95,-1.4) and (3.31,-0.3) .. (0,0) .. controls (3.31,0.3) and (6.95,1.4) .. (10.93,3.29)   ;
\draw  [dash pattern={on 4.5pt off 4.5pt}]  (515,152) -- (559,152) ;
\draw [shift={(561,152)}, rotate = 180] [color={rgb, 255:red, 0; green, 0; blue, 0 }  ][line width=0.75]    (10.93,-3.29) .. controls (6.95,-1.4) and (3.31,-0.3) .. (0,0) .. controls (3.31,0.3) and (6.95,1.4) .. (10.93,3.29)   ;

\draw (562,205.9) node [anchor=north west][inner sep=0.75pt]    {$r$};
\draw (394,141.4) node [anchor=north west][inner sep=0.75pt]    {Probability: $\varepsilon$};
\draw (310,213.4) node [anchor=north west][inner sep=0.75pt]    {$\mathrm{VaR}_{1-\varepsilon}$};
\draw (380,213.4) node [anchor=north west][inner sep=0.75pt]    {$\mathrm{CVaR}_{1-\varepsilon}$};
\draw (238,213.4) node [anchor=north west][inner sep=0.75pt]    {$\mathbb{E}[R]$};
\draw (163,131.4) node [anchor=north west][inner sep=0.75pt]    {$\mathrm{f}_R$};
\end{tikzpicture}

%% file: DRSTL.bbl
\begin{thebibliography}{10}
\providecommand{\url}[1]{#1}
\csname url@samestyle\endcsname
\providecommand{\newblock}{\relax}
\providecommand{\bibinfo}[2]{#2}
\providecommand{\BIBentrySTDinterwordspacing}{\spaceskip=0pt\relax}
\providecommand{\BIBentryALTinterwordstretchfactor}{4}
\providecommand{\BIBentryALTinterwordspacing}{\spaceskip=\fontdimen2\font plus
\BIBentryALTinterwordstretchfactor\fontdimen3\font minus \fontdimen4\font\relax}
\providecommand{\BIBforeignlanguage}[2]{{%
\expandafter\ifx\csname l@#1\endcsname\relax
\typeout{** WARNING: IEEEtran.bst: No hyphenation pattern has been}%
\typeout{** loaded for the language `#1'. Using the pattern for}%
\typeout{** the default language instead.}%
\else
\language=\csname l@#1\endcsname
\fi
#2}}
\providecommand{\BIBdecl}{\relax}
\BIBdecl

\bibitem{dennis2023verifiable}
L.~A. Dennis and M.~Fisher, \emph{Verifiable Autonomous Systems: Using Rational Agents to Provide Assurance about Decisions Made by Machines}.\hskip 1em plus 0.5em minus 0.4em\relax Cambridge University Press, 2023.

\bibitem{lindemann2025formal}
L.~Lindemann and D.~V. Dimarogonas, \emph{Formal Methods for Multi-agent Feedback Control Systems}.\hskip 1em plus 0.5em minus 0.4em\relax MIT Press, 2025.

\bibitem{maler2004monitoring}
O.~Maler and D.~Nickovic, ``Monitoring temporal properties of continuous signals,'' in \emph{International Symposium on Formal Techniques in Real-Time and Fault-Tolerant Systems}.\hskip 1em plus 0.5em minus 0.4em\relax Springer, 2004, pp. 152--166.

\bibitem{donze2010robust}
A.~Donz{\'e} and O.~Maler, ``Robust satisfaction of temporal logic over real-valued signals,'' in \emph{International Conference on Formal Modeling and Analysis of Timed Systems}.\hskip 1em plus 0.5em minus 0.4em\relax Springer, 2010, pp. 92--106.

\bibitem{raman-cdc14}
V.~Raman, A.~Donz{\'{e}}, M.~Maasoumy, R.~M. Murray, A.~Sangiovanni-Vincentelli, and S.~A. Seshia, ``Model predictive control with signal temporal logic specifications,'' in \emph{Proceedings of the 53rd IEEE Conference on Decision and Control (CDC)}, December 2014, pp. 81--87.

\bibitem{farahani2017shrinking}
S.~S. Farahani, R.~Majumdar, V.~S. Prabhu, and S.~E.~Z. Soudjani, ``Shrinking horizon model predictive control with chance-constrained signal temporal logic specifications,'' in \emph{2017 American Control Conference (ACC)}.\hskip 1em plus 0.5em minus 0.4em\relax IEEE, 2017, pp. 1740--1746.

\bibitem{Haesaert2018}
S.~Haesaert, P.~Nilsson, C.~Vasile, R.~Thakker, A.~Agha-mohammadi, A.~Ames, and R.~Murray, ``Temporal logic control of pomdps via label-based stochastic simulation relations,'' \emph{IFAC-PapersOnLine}, vol.~51, no.~16, pp. 271--276, 2018.

\bibitem{Safaoui2020}
S.~Safaoui, L.~Lindemann, D.~V. Dimarogonas, I.~Shames, and T.~H. Summers, ``{Control Design for Risk-Based Signal Temporal Logic Specifications},'' \emph{IEEE Control Systems Letters}, vol.~4, no.~4, pp. 1000--1005, 2020.

\bibitem{LindemannCDC2020}
L.~Lindemann, G.~J. Pappas, and D.~V. Dimarogonas, ``{Control Barrier Functions for Nonholonomic Systems under Risk Signal Temporal Logic Specifications},'' in \emph{Proceedings of the IEEE Conference on Decision and Control}, 2020, pp. 1422--1428.

\bibitem{Sadigh2016}
D.~Sadigh and A.~Kapoor, ``{Safe Control under Uncertainty with Probabilistic Signal Temporal Logic},'' in \emph{Proceedings of Robotics: Science and Systems}, jun 2016.

\bibitem{Sadigh2018}
S.~Jha, V.~Raman, D.~Sadigh, and S.~A. Seshia, ``{Safe Autonomy Under Perception Uncertainty Using Chance-Constrained Temporal Logic},'' \emph{Journal of Automated Reasoning}, vol.~60, no.~1, pp. 43--62, 2018.

\bibitem{Li2017}
J.~Li, P.~Nuzzo, A.~Sangiovanni-Vincentelli, Y.~Xi, and D.~Li, ``{Stochastic contracts for cyber-physical system design under probabilistic requirements},'' in \emph{MEMOCODE 2017 - 15th ACM-IEEE International Conference on Formal Methods and Models for System Design}, 2017, pp. 5--14.

\bibitem{Kyriakis2019}
P.~Kyriakis, J.~V. Deshmukh, and P.~Bogdan, ``{Specification mining and robust design under uncertainty: A stochastic temporal logic approach},'' \emph{ACM Transactions on Embedded Computing Systems}, vol.~18, no.~5s, 2019.

\bibitem{Tiger2020}
M.~Tiger and F.~Heintz, ``{Incremental reasoning in probabilistic Signal Temporal Logic},'' \emph{International Journal of Approximate Reasoning}, vol. 119, pp. 325--352, 2020.

\bibitem{Ilyes2023}
R.~B. Ilyes, Q.~H. Ho, and M.~Lahijanian, ``Stochastic robustness interval for motion planning with signal temporal logic,'' in \emph{2023 IEEE International Conference on Robotics and Automation (ICRA)}, 2023, pp. 5716--5722.

\bibitem{ScherHSCC2022}
G.~Scher, S.~Sadraddini, R.~Tedrake, and H.~Kress-Gazit, ``Elliptical slice sampling for probabilistic verification of stochastic systems with signal temporal logic specifications,'' in \emph{Proceedings of the 25th ACM International Conference on Hybrid Systems: Computation and Control}, ser. HSCC '22.\hskip 1em plus 0.5em minus 0.4em\relax New York, NY, USA: Association for Computing Machinery, 2022.

\bibitem{Scher2022}
G.~Scher, S.~Sadraddini, and H.~Kress-Gazit, ``Robustness-based synthesis for stochastic systems under signal temporal logic tasks,'' in \emph{2022 IEEE/RSJ International Conference on Intelligent Robots and Systems (IROS)}, 2022, pp. 1269--1275.

\bibitem{FarahaniTAC2019}
S.~S. Farahani, R.~Majumdar, V.~S. Prabhu, and S.~Soudjani, ``{Shrinking Horizon Model Predictive Control With Signal Temporal Logic Constraints Under Stochastic Disturbances},'' \emph{IEEE Transactions on Automatic Control}, vol.~64, no.~8, pp. 3324--3331, 2019.

\bibitem{SadBel:15}
S.~Sadraddini and C.~Belta, ``Robust temporal logic model predictive control,'' in \emph{Allerton Conference}, 2015.

\bibitem{Farahani2015}
S.~S. Farahani, V.~Raman, and R.~M. Murray, ``Robust model predictive control for signal temporal logic synthesis,'' \emph{IFAC-PapersOnLine}, vol.~48, no.~27, pp. 323--328, 2015, analysis and Design of Hybrid Systems ADHS.

\bibitem{RamanHSCC2015}
V.~Raman, A.~Donz\'{e}, D.~Sadigh, R.~M. Murray, and S.~A. Seshia, ``Reactive synthesis from signal temporal logic specifications,'' in \emph{Proceedings of the 18th International Conference on Hybrid Systems: Computation and Control}, ser. HSCC '15.\hskip 1em plus 0.5em minus 0.4em\relax New York, NY, USA: Association for Computing Machinery, 2015, p. 239–248.

\bibitem{Campi2011}
M.~C. Campi and S.~Garatti, ``{A Sampling-and-Discarding Approach to Chance-Constrained Optimization: Feasibility and Optimality},'' \emph{Journal of Optimization Theory and Applications}, vol. 148, no.~2, pp. 257--280, 2011.

\bibitem{Calafiore2010}
G.~C. Calafiore, ``Random convex programs,'' \emph{SIAM Journal on Optimization}, vol.~20, no.~6, pp. 3427--3464, 2010.

\bibitem{Grammatico2016}
S.~Grammatico, X.~Zhang, K.~Margellos, P.~Goulart, and J.~Lygeros, ``A scenario approach for non-convex control design,'' \emph{IEEE Transactions on Automatic Control}, vol.~61, no.~2, pp. 334--345, 2016.

\bibitem{delage2010distributionally}
E.~Delage and Y.~Ye, ``Distributionally robust optimization under moment uncertainty with application to data-driven problems,'' \emph{Operations research}, vol.~58, no.~3, pp. 595--612, 2010.

\bibitem{hu2013kullback}
Z.~Hu and L.~J. Hong, ``Kullback-leibler divergence constrained distributionally robust optimization,'' \emph{Available at Optimization Online}, pp. 1695--1724, 2013.

\bibitem{pflug2007ambiguity}
G.~Pflug and D.~Wozabal, ``Ambiguity in portfolio selection,'' \emph{Quantitative Finance}, vol.~7, no.~4, pp. 435--442, 2007.

\bibitem{mohajerin2018data}
P.~Mohajerin~Esfahani and D.~Kuhn, ``Data-driven distributionally robust optimization using the {W}asserstein metric: Performance guarantees and tractable reformulations,'' \emph{Mathematical Programming}, vol. 171, no.~1, pp. 115--166, 2018.

\bibitem{kordabad2022safe}
A.~B. Kordabad, R.~Wisniewski, and S.~Gros, ``Safe reinforcement learning using {W}asserstein distributionally robust {MPC} and chance constraint,'' \emph{IEEE Access}, vol.~10, pp. 130\,058--130\,067, 2022.

\bibitem{gao2023distributionally}
R.~Gao and A.~Kleywegt, ``Distributionally robust stochastic optimization with {W}asserstein distance,'' \emph{Mathematics of Operations Research}, vol.~48, no.~2, pp. 603--655, 2023.

\bibitem{chen2024data}
Z.~Chen, D.~Kuhn, and W.~Wiesemann, ``Data-driven chance constrained programs over {W}asserstein balls,'' \emph{Operations Research}, vol.~72, no.~1, pp. 410--424, 2024.

\bibitem{xie2021distributionally}
W.~Xie, ``On distributionally robust chance constrained programs with {W}asserstein distance,'' \emph{Mathematical Programming}, vol. 186, no.~1, pp. 115--155, 2021.

\bibitem{hota2019data}
A.~R. Hota, A.~Cherukuri, and J.~Lygeros, ``Data-driven chance constrained optimization under {W}asserstein ambiguity sets,'' in \emph{2019 American Control Conference (ACC)}.\hskip 1em plus 0.5em minus 0.4em\relax IEEE, 2019, pp. 1501--1506.

\bibitem{Soudjani2018}
S.~Soudjani and R.~Majumdar, ``Concentration of measure for chance-constrained optimization,'' \emph{IFAC-PapersOnLine}, vol.~51, no.~16, pp. 277--282, 2018, 6th IFAC Conference on Analysis and Design of Hybrid Systems ADHS 2018.

\bibitem{Ledoux1999}
M.~Ledoux, ``\BIBforeignlanguage{eng}{Concentration of measure and logarithmic {S}obolev inequalities},'' \emph{\BIBforeignlanguage{eng}{Séminaire de probabilités de Strasbourg}}, vol.~33, pp. 120--216, 1999.

\bibitem{rockafellar2000optimization}
R.~T. Rockafellar, S.~Uryasev \emph{et~al.}, ``Optimization of conditional value-at-risk,'' \emph{Journal of risk}, vol.~2, pp. 21--42, 2000.

\bibitem{lindemann2023safe}
L.~Lindemann, M.~Cleaveland, G.~Shim, and G.~J. Pappas, ``Safe planning in dynamic environments using conformal prediction,'' \emph{IEEE Robotics and Automation Letters}, 2023.

\bibitem{yu2023signal}
X.~Yu, Y.~Zhao, X.~Yin, and L.~Lindemann, ``Signal temporal logic control synthesis among uncontrollable dynamic agents with conformal prediction,'' \emph{arXiv preprint arXiv:2312.04242}, 2024.

\bibitem{kordabad2024distributionally}
A.~B. Kordabad, E.~E. Vlahakis, L.~Lindemann, D.~V. Dimarogonas, and S.~Soudjani, ``Distributionally robust control for chance-constrained signal temporal logic specifications,'' \emph{arXiv preprint arXiv:2409.03855}, 2024.

\bibitem{MalNic:04}
O.~Maler and D.~Nickovic, ``Monitoring temporal properties of continuous signals,'' in \emph{{FORMATS/FTRTFT}}, ser. LNCS 3253.\hskip 1em plus 0.5em minus 0.4em\relax Springer, 2004.

\bibitem{req_mining_hscc2013}
X.~Jin, A.~Donz{\'{e}}, J.~V. Deshmukh, and S.~A. Seshia, ``Mining requirements from closed-loop control models,'' \emph{{IEEE} Trans. on {CAD} of Integrated Circuits and Systems}, vol.~34, no.~11, pp. 1704--1717, 2015.

\bibitem{ledoux2001concentration}
M.~Ledoux, \emph{The concentration of measure phenomenon}.\hskip 1em plus 0.5em minus 0.4em\relax American Mathematical Soc., 2001, no.~89.

\bibitem{barvinok1997measure}
A.~Barvinok, ``Measure concentration in optimization,'' \emph{Mathematical Programming}, vol.~79, pp. 33--53, 1997.

\bibitem{pisier1999volume}
G.~Pisier, \emph{The volume of convex bodies and Banach space geometry}.\hskip 1em plus 0.5em minus 0.4em\relax Cambridge University Press, 1999, vol.~94.

\bibitem{kordabad2024control}
A.~B. Kordabad, M.~Charitidou, D.~V. Dimarogonas, and S.~Soudjani, ``Control barrier functions for stochastic systems under signal temporal logic tasks,'' in \emph{2024 European Control Conference (ECC)}.\hskip 1em plus 0.5em minus 0.4em\relax IEEE, 2024, pp. 3213--3219.

\bibitem{rockafellar2002conditional}
R.~T. Rockafellar and S.~Uryasev, ``Conditional value-at-risk for general loss distributions,'' \emph{Journal of banking \& finance}, vol.~26, no.~7, pp. 1443--1471, 2002.

\bibitem{fournier2015rate}
N.~Fournier and A.~Guillin, ``On the rate of convergence in {W}asserstein distance of the empirical measure,'' \emph{Probability Theory and Related Fields}, vol. 162, no.~3, pp. 707--738, 2015.

\bibitem{gilpin2020smooth}
Y.~Gilpin, V.~Kurtz, and H.~Lin, ``A smooth robustness measure of signal temporal logic for symbolic control,'' \emph{IEEE Control Systems Letters}, vol.~5, no.~1, pp. 241--246, 2020.

\bibitem{kazemi2020formal}
M.~Kazemi and S.~Soudjani, ``Formal policy synthesis for continuous-state systems via reinforcement learning,'' in \emph{Integrated Formal Methods: 16th International Conference, IFM 2020, Lugano, Switzerland, November 16--20, 2020, Proceedings 16}.\hskip 1em plus 0.5em minus 0.4em\relax Springer, 2020, pp. 3--21.

\end{thebibliography}
